\documentclass[final,a4paper,3p]{elsarticle}
\pdfoutput=1
\usepackage{hypernat}
\usepackage[colorlinks,hypertexnames=false]{hyperref}

\usepackage[english]{babel}
\biboptions{authoryear}

\usepackage{amsmath}
\usepackage{amsfonts}
\usepackage{booktabs}
\usepackage{stmaryrd}
\usepackage{multirow}
\usepackage{algorithm}
\usepackage{algpseudocode}
\usepackage[obeyDraft]{todonotes}
\usepackage{placeins}
\usepackage{enumitem}
\usetikzlibrary{calc}

\newcommand{\tens}[1]{{\boldsymbol{#1}}}
\newcommand{\ftens}[1]{{\widehat{\boldsymbol{#1}}}}
\renewcommand{\d}[0]{{\mathrm{d}}}
\renewcommand{\epsilon}{\varepsilon}
\renewcommand{\div}{\mathrm{div}}
\newcommand{\hooke}[0]{{\tens{\mathcal{C}}}}
\newcommand{\navier}[0]{{\tens{\mathrm{N}}}}
\newcommand{\fnavier}[0]{{\ftens{\mathrm{N}}}}
\newcommand{\Id}[0]{{\tens{I}}}
\newcommand{\body}[0]{{\mathcal{B}}}
\newcommand{\dual}[3]{{\left\langle #1, #2\right\rangle_{#3}}}
\newcommand{\smallbullet}{\raisebox{0.5pt}{\scalebox{0.6}{$\bullet$}}}
\DeclareMathOperator\Tr{Tr}

\newtheorem{theorem}{Theorem}
\newtheorem{lemma}[theorem]{Lemma}
\newproof{proof}{Proof}
\newdefinition{remark}{Remark}

\newcommand{\MyComment}[1]{{\Comment{\textcolor{gray}{\texttt{#1}}}}}

\usepackage{autonum}

\graphicspath{{figs/}}

\newcommand{\hsub}{_{\text{\scriptsize h}}}
\newcommand{\ysub}{_{\text{\scriptsize y}}}
\newcommand{\Cbb} {\mathbb{C}}
\newcommand{\Rbb} {\mathbb{R}}
\newcommand{\Zbb} {\mathbb{Z}}

\newcommand{\Mcal}{\mathcal{M}}
\newcommand{\Ncal}{\mathcal{N}}
\newcommand{\Bcal}{\mathcal{B}}
\newcommand{\Ocal}{\mathcal{O}}
\newcommand{\Rcal}{\mathcal{R}}

\newcommand{\dy}{\,\text{d}y}
\newcommand{\dz}{\,\text{d}z}
\newcommand{\dyt}{\,\text{d}\bfyt}
\newcommand{\bfxt} {\tilde{\boldsymbol{x}}{}}
\newcommand{\bfyt} {\tilde{\boldsymbol{y}}{}}
\newcommand{\bfzt} {\tilde{\boldsymbol{z}}{}}
\newcommand{\rmi}{\mathrm{i}}
\newcommand{\Fcal}{\mathcal{F}}
\newcommand{\Ecal}{\mathcal{E}}
\newcommand{\Scal}{\mathcal{S}}
\newcommand{\shm}{\hspace*{-0.1em}-\hspace*{-0.1em}}
\newcommand{\shp}{\hspace*{-0.1em}+\hspace*{-0.1em}}

\begin{document}
\title{A Fourier-accelerated volume integral method for elastoplastic contact}

\author[epfl]{Lucas Fr\'erot\corref{cor}}
\ead{lucas.frerot@epfl.ch}

\author[ensta]{Marc Bonnet}
\ead{marc.bonnet@ensta-paristech.fr}

\author[epfl]{Jean-François Molinari}
\ead{jean-francois.molinari@epfl.ch}

\author[epfl]{Guillaume Anciaux}
\ead{guillaume.anciaux@epfl.ch}

\address[epfl]{Institute of Civil Engineering, École Polytechnique Fédérale de Lausanne (EPFL), CH 1015 Lausanne, Switzerland}

\address[ensta]{POEMS (UMR 7231 CNRS-ENSTA-INRIA), ENSTA Paristech, France}

\cortext[cor]{Corresponding Author}


\begin{abstract}
  The contact of solids with rough surfaces plays a fundamental role
  in physical phenomena such as friction, wear, sealing, and thermal
  transfer. However, its simulation is a challenging problem due to
  surface asperities covering a wide range of length-scales. In
  addition, non-linear local processes, such as plasticity, are
  expected to occur even at the lightest loads. In this context,
  robust and efficient computational approaches are required. We
  therefore present a novel numerical method, based on integral
  equations, capable of handling the large discretization requirements
  of real rough surfaces as well as the non-linear plastic flow
  occurring below and at the contacting asperities. This method is
  based on a new derivation of the Mindlin fundamental solution in
  Fourier space, which leverages the computational efficiency of the
  fast Fourier transform. The use of this Mindlin solution allows a
  dramatic reduction of the memory imprint (as the Fourier
  coefficients are computed on-the-fly), a reduction of the
  discretization error, and the exploitation of the structure of the
  functions to speed up computation of the integral operators. We
  validate our method against an elastic-plastic FEM Hertz normal
  contact simulation and showcase its ability to simulate contact of
  rough surfaces with plastic flow.
\end{abstract}



\maketitle

Tribological phenomena, such as friction and wear, have a significant
impact on the behavior of natural and man-made systems, from the
kilometer scale in the case of earthquakes, down to the nanometer in
nanotribology. In 1939, Bowden and Tabor showed that friction is
intimately linked to the true contact area of mating
surfaces~\citep{bowden_area_1939}, while Archard made the same
statement for adhesive wear in 1953~\citep{archard_contact_1953}. It
is with the goal of quantitatively studying the physics at the true
contact area that various methods for rough surface contact have been
developed. The seminal work of~\citet{greenwood_contact_1966} paved
the way for multi-asperity contact via statistical methods, which
evolved over time to include asperity
curvature~\citep{bush_elastic_1975}, long-range elastic
interaction~\citep{vergne_elastic_1985}, surface
anisotropy~\citep{bush_strongly_1979}, surface power spectral
density~\citep{persson_contact_2006}, etc. Numerical methods have also
been developed and used to study rough surface contact. Although
finite elements methods have been successfully
employed~\citep{hyun_finiteelement_2004, yastrebov_rough_2011},
boundary integral methods have proven very efficient in dealing with
elastic rough contact. These include spectral
methods~\citep{stanley_fftbased_1997, yastrebov_contact_2012,
  rey_normal_2017} and Green's functions
methods~\citep{polonsky_numerical_1999, campana_contact_2007,
  putignano_new_2012}. They have been used to study the true contact
area evolution in adhesionless~\citep{yastrebov_infinitesimal_2015}
and adhesive contact~\citep{carbone_adhesive_2009,
  pastewka_contact_2014, rey_normal_2017}, as well as interfacial
separation~\citep{almqvist_interfacial_2011}, the autocorrelation of
the surface stresses and micro-contacts~\citep{campana_elastic_2008,
  ramisetti_autocorrelation_2011}, the distribution of the areas of
micro-contacts~\citep{frerot_mechanistic_2018,
  muser_contactpatchsize_2018}, the autocorrelation of sub-surface
stresses~\citep{muser_internal_2018}, etc. The contact mechanics
challenge organized by Müser in 2016 references as many as 13
different techniques for adhesive elastic contact, from 12 different
research groups, studying some of the quantities previously
mentioned~\citep{muser_meeting_2017}.

However, elastic theories have difficulties in representing realistic
contact behavior: local contact pressures can easily reach values
higher than the Young's modulus~\citep{muser_internal_2018}, and the
true contact area vanishes as the rough surface spectrum grows
wider~\citep{persson_contact_2006}. The works of
\citet{bowden_area_1939}, \citet{archard_contact_1953}, and more
recently \citet{weber_molecular_2018} have experimentally shown that
elastic theories are not able to model realistic contact
interfaces. It is the consensus of the tribology
community~\citep{vakis_modeling_2018} that contact models need to
evolve to include material non-linearities such as plasticity.

The work of~\citet*{pei_finite_2005}, while pioneering the study of
elastic-plastic rough surface contact (using a finite element
approach), suffers from discretization error (a single element is used
to reproduce the smallest surface wavelength), inaccuracy of an
explicit dynamic relaxation scheme to reduce the simulation cost of a
static calculation, as well as statistical errors from the surface
spectrum choice~\citep{yastrebov_contact_2012} and number of
realizations. Pei et al.\ nonetheless confirmed the key role of
plasticity in quantifying the true contact area and the micro-contact
distribution. Consequently, the objective of our paper is to propose a
high-performance, robust and quantitatively accurate method to study
the contact behavior of elastic-plastic materials with statistically
representative rough surfaces.\enlargethispage*{3ex}

Our method is based on a volume integral
approach~\citep{telles_application_1979, telles_implicit_1991,
  bonnet_implicit_1996}. Similarly to the boundary methods previously
mentioned, the volume integral methods (VIMs) can exactly represent the
elastic behavior of a semi-infinite solid. This limits the volume
discretization to the potentially plastic regions, allowing better
usage of computational resources. VIMs however require knowledge of
singular fundamental solutions and the computation of their volume
convolution with plastic deformations. This operation is costly with
a ``naive'' implementation~\citep{jacq_development_2002}, its algorithmic complexity being $\Ocal(N^2)$ ($N$ being the total number of plastic deformation unknowns). It has been accelerated using 2D fast-Fourier transform (FFT) of the fundamental
solutions~\citep{sainsot_numerical_2002, wang_numerical_2005}, at the
cost of introducing a sampling error in the computed quantities of
interest. Another approach for accelerating the convolution computation
is to use a 3D-FFT~\citep{wang_new_2013}, but this introduces a
periodicity error and requires discretization of a domain more than
twice the volume of the expected plastic zone to reduce these effects,
thereby reducing the attractivity of VIMs for semi-infinite modeling.

Accordingly, we present a volume integral method well suited for
periodic elastic-plastic contact problems that is based on a novel
derivation of the half-space fundamental solutions \emph{directly} in
the 2D partial Fourier domain. The advantage is three-fold: it allows
the use of the 2D-FFT to speed up the convolution computation in the
plane parallel to the contact surface \emph{without introducing
  sampling error}, it \emph{does not require storage} of the discrete
Fourier coefficients of the real-space fundamental solutions, and
permits \emph{optimization of the convolution computation} by
exploiting the structure of the analytical solutions, while keeping
the advantages described above for volume integral methods. Indeed, the algorithmic complexity of our treatment is only $\Ocal(N\log(N))$ per convolution. Moreover, the use of the FFT renders the method matrix-free, making it attractive for use
in conjunction with iterative solvers.

After stating the periodic elastic-plastic contact problem in
Section~\ref{sec:problem_statement}, we derive in
Section~\ref{sec:fundamental} the fundamental solutions needed for the
integral operators of the method and their application in a
discretized periodic context
(Section~\ref{sec:discretized_operators}). We present in
Section~\ref{sec:implicit} the implicit volume integral equation for
the plasticity problem, and its iterative coupling with the elastic
contact problem in Section~\ref{sec:coupling}. We subsequently show
validation cases for the integral operators and the global
elastic-plastic contact method (Section~\ref{sec:validation}) and the algorithmic complexity of the integral operators
(Section~\ref{sec:complexity}), ending with an application of the
method to rough surface contact in Section~\ref{sec:rough}.


\section{Problem statement and overview of solution
  methodology}\label{sec:problem_statement}

Let
$\body := \{\tens{y} \in \Rbb^3 : \tens{y}\cdot\tens{e}_3 \geq 0\}$ be
a deformable semi-infinite elastic-plastic solid of boundary
$\partial\body$, see figure~\ref{fig:schematic_presentation}. Points
$\tens{y}\in\body$ will often be denoted as $\tens{y}=(\bfyt,y_3)$
with $\bfyt:=(y_1,y_2)$. Let $\tens{u}$ be the displacement vector
field of $\body$. The linearized strain tensor $\tens{\epsilon}$ and
the Cauchy stress tensor $\tens{\sigma}$ are respectively given by:
\begin{align}
  \tens{\epsilon}[\tens{u}] & := \frac{1}{2}\left(\tens{\nabla u} + \tens{\nabla u}^T\right),\label{eqn:epsilon}\\
  \tens{\sigma}[\tens{u}, \tens{\epsilon}^p] & := \hooke:\big(\tens{\epsilon}[\tens{u}] - \tens{\epsilon}^p\big),\label{eqn:sigma}
\end{align}
where $\hooke \in \Rbb^{3\times3\times3\times3}$ is the elasticity
tensor, satisfying the usual ellipticity and (major and minor)
symmetries, while $\tens{\epsilon}^p$ is the plastic strain. Here and
thereafter, we follow the usual convention whereby the gradient
operator $\tens{\nabla}$ acts ``to the right'', so that
e.g. $(\tens{\nabla u})_{ij}=\partial_j u_i=u_{i,j}$.  The stress
$\tens{\sigma}$ satisfies the conservation of momentum equation in the
absence of body forces:
\begin{equation}
  \label{eqn:equilibrium}\div\tens{\sigma} = \tens{0}\quad\mathrm{a.e.\ in}\ \body.
\end{equation}
Let $\tens{n} := -\tens{e}_3$ be the external normal of
$\partial \body$, we define the traction vector of the displacement
field $\tens{u}$ as:
\begin{equation}
  \label{eqn:traction}\tens{T}[\tens{u}, \tens{\epsilon}^p] := \tens{\sigma}[\tens{u}, \tens{\epsilon}^p]\Big|_{\partial\body}\cdot \tens{n}.
\end{equation}

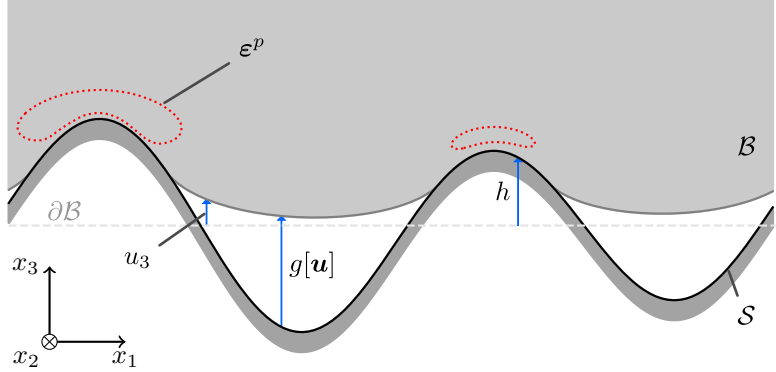
\begin{figure}
  \centering
  \begin{tikzpicture}
    \node(0, 0) {\includegraphics[draft=False]{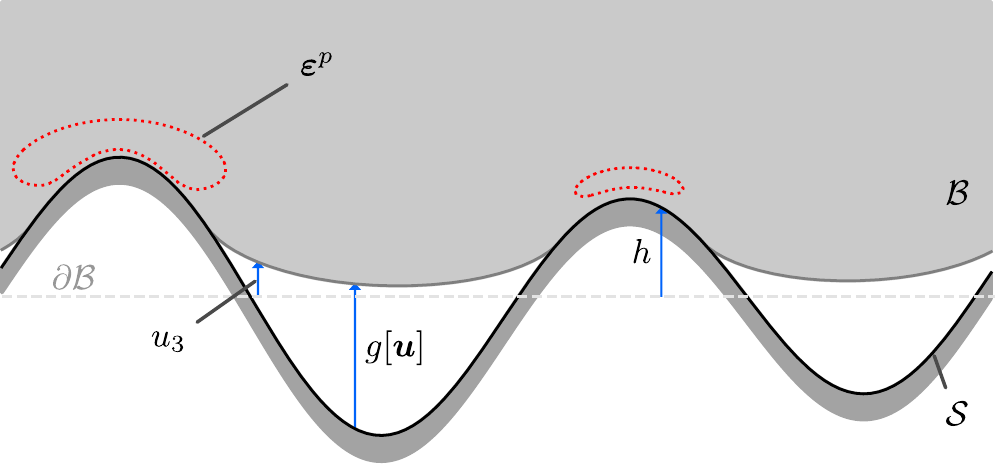}};
    \def\x{-4.5}; \def\y{-2.2};

    \draw[->, thick] (\x, \y) -- ++(1, 0) node[below]{$x_1$};
    \draw[->, thick] (\x, \y) -- ++(0, 1) node[left]{$x_3$};
    \node[below left] at (\x, \y){$x_2$}; \filldraw[fill=white] (\x,
    \y) circle[radius=0.1]; \draw ({\x + sqrt(2)/2*0.1}, {\y +
      sqrt(2)/2*0.1}) -- ({\x - sqrt(2)/2*0.1}, {\y - sqrt(2)/2*0.1});
    \draw ({\x - sqrt(2)/2*0.1}, {\y + sqrt(2)/2*0.1}) -- ({\x +
      sqrt(2)/2*0.1}, {\y - sqrt(2)/2*0.1});
  \end{tikzpicture}
  \caption{\textbf{Schematic representation of a periodic
      elastic-plastic contact problem}, with geometrical contact
    quantities represented by blue arrows and plastic deformation
    zones with dashed red contours. The elastic-plastic body $\body$
    is represented in deformed shape. Note that this schematic is
    taken from a real two-dimensional
    simulation.}\label{fig:schematic_presentation}
\end{figure}

The evolution of the plastic strain is, for definiteness, assumed to
obey a standard $J_2$ plasticity model with an associated flow rule
(see Section~\ref{sec:von_mises} for more details).

\subsection{Contact boundary conditions}

We consider situations where the boundary $\partial \body$ is in
contact with a surface $\Scal$, which is defined as the graph of a
scalar function $h \in C^0(\partial\body)$ and can be viewed as the
boundary of an infinitely stiff solid. The resulting normal force
bringing $\body$ and $\Scal$ together is denoted
$\tens{W} := W\tens{e}_3$. We ignore friction and adhesion between
$\body$ and $\Scal$ and define the gap function as:
\begin{equation}
  g[\tens{u}] := \tens{u}\cdot\tens{e}_3\Big|_{\partial\body} - h,
\end{equation}
which is the separation between the deformed solid $\body$ and the
surface $\Scal$ along $\tens{n}$. The boundary conditions on
$\partial\body$ are hence expressed with the Hertz-Signorini-Moreau
conditions:
\begin{subequations}
  \begin{align}
    g[\tens{u}] \geq 0,\\
    p[\tens{u}, \tens{\epsilon}^p] := \tens{T}[\tens{u}, \tens{\epsilon}^p]\cdot \tens{e}_3 \geq 0,\\
    g[\tens{u}]p[\tens{u}, \tens{\epsilon}^p] = 0.
  \end{align}
\end{subequations}
Moreover, for reasons discussed in Section~\ref{sec:elastic_contact},
we treat the contacting surface $\Scal$ as fixed, and enforce the
applied resulting force $W$ as an additional constraint that will
allow to determine the mean vertical displacement:
\begin{equation}
  \label{eqn:total_pressure}\int_{\partial\body}{p[\tens{u}, \tens{\epsilon}^p]\,\mathrm{d} A} = W.
\end{equation}

\subsection{Horizontally-periodic setting}
As we will employ Fourier methods to solve the elastic-plastic contact
problem, we set the latter in a more natural setting for the
application of the discrete Fourier transform. Let us define the
periodic cell
\begin{align}
  \label{eqn:periodic_cell}\Bcal_p & = \left\rbrack -\frac{L_1}{2}, \frac{L_1}{2} \right\lbrack \times \left\rbrack -\frac{L_2}{2}, \frac{L_2}{2} \right\lbrack \times \mathbb{R}^+,
\end{align}
where $L_1$ (resp.\ $L_2$) is the horizontal dimension of the cell in
the direction $\tens{e}_1$ (resp.\ $\tens{e}_2$). It will become
apparent in Section~\ref{sec:discretized_operators} that this helps
defining the discretization procedure of the continuous operators
presented below, since $\body_p$-periodic fields admit a
representation as Fourier series.

\subsection{Overview of solution methodology}
In this work, we use an integral equation approach for solving the
elastic-plastic contact problem. We take advantage of the well-known
fact that the elastostatic displacement $\tens{u}$ generated in
$\body$ or on $\partial\body$ by any given traction distribution
$\tens{p}$ on $\partial\body$ and eigenstress distribution $\tens{w}$
in $\body$ has the explicit representation
\begin{equation}\label{eqn:operators}
  \tens{u} = \Mcal[\tens{p}] + \Ncal[\tens{w}],
\end{equation}
where the integral operators $\Mcal, \Ncal$ are defined in terms of
elastostatic fundamental solutions. This framework accounts for the
semi-infinite geometry of the problem, which a classical
finite-element method does not, and enforces automatic satisfaction of
equilibrium and elastic constitutive relations. As a result, the
solution of contact and elastic-plastic problems only requires the
satisfaction of relevant relations on the contact surface (using
evaluations of~\eqref{eqn:operators} at surface points) and plastic
regions (exploiting $\tens{\nabla u}$ at internal points),
respectively, with the latter task entailing computations of
$\tens{\nabla}\Mcal[\tens{p}]$ and $\tens{\nabla}\Ncal[\tens{w}]$ for
given $\tens{p}$ or $\tens{w}$.


In addition, the semi-infinite geometry and constitutive uniformity
assumed in this work together imply translational invariance in any
horizontal direction. For example, for any point $\tens{x}\in\body$,
the displacements $\Ncal[\tens{w}](\tens{x})$ and
$\Ncal[\tens{w}(\bullet+\tens{h})](\tens{x}+\tens{h})$ are equal for
any horizontal translation vector
$\tens{h}=(h_1,h_2,0)$. Consequently, the integral operators
$\Mcal,\Ncal$ can be expressed as convolutions w.r.t.\ the horizontal
coordinates; for example there is a tensor-valued function
$\tens{H}(\bfzt,x_3,y_3)$, classically defined for an aperiodic
infinite domain~\citep{fredholm_equations_1900}, such that
\begin{equation}
  \Ncal[\tens{w}](\tens{x})
  = \int_{0}^{\infty} \Big\{ \int_{\Rbb^2} \tens{H}(\bfyt-\bfxt,x_3,y_3):\tens{w}(\tens{y}) \dyt \Big\} \dy_3 \label{generique:N}
\end{equation}
(see Sec.~\ref{mindlin}). We accordingly define the partial
convolution operation along the horizontal coordinates
$\bfxt=(x_1,x_2)$ by
\begin{equation}
  (f\star g)(\bfxt) := \int_{\Rbb^2} f(\bfyt-\bfxt)g(\bfyt) \dyt. \label{eqn:convolutions}
\end{equation}
where the operands $f,g$ may in addition depend on the vertical
coordinates $x_3,y_3$; moreover, the convolution symbol $\star$ will
implicitly retain any tensor contractions involved in the integrals
being recast in convolution form. For example,
reformulating~\eqref{generique:N} using
definition~\eqref{eqn:convolutions} gives
\begin{equation}
  \Ncal[\tens{w}](\tens{x})
  = \int_{0}^{\infty} (\tens{H}\star\tens{w})(\bfxt,x_3,y_3) \dy_3 \label{generique:N:conv}
\end{equation}
\begin{remark}\label{convolvable}
  We emphasize that the partial convolution~\eqref{eqn:convolutions}
  makes sense only provided its operands $f,g$ are ``convolvable'',
  i.e.\ are in some sense mutually compatible. It is in particular
  well-defined in the classical sense if $f,g$ are either (a)
  integrable functions over $\Rbb^2$ (i.e. $f,g\in L^1(\Rbb^2)$) or
  (b) locally integrable functions whose supports are
  convolutive\footnote{Meaning that $f$ and $g$ have supports such
    that $\bfyt-\bfxt$ remains bounded, which is in particular the
    case when one of the operands is compactly supported.}, and only
  in the sense of distributions, a.k.a.\ generalized functions, or if
  (c) $f$ is a tempered distribution and $g$ is a compactly supported
  distribution\footnote{i.e. $f\in\Scal'$ and $g\in\Ecal'$, using
    standard notation for spaces defined by the theory of
    distributions, see e.g. the appendix ``Distributions''
    in~\cite{dautray_mathematical_2000}.}. We will shortly see that
  such considerations play an important role in this work, which
  involves convolutions of types (b) or (c) as kernels $\tens{H}$
  derived from fundamental solutions may be locally-integrable but are
  not integrable.
\end{remark}
\begin{remark}\label{conv:def:rem}
  Notice the slightly unconventional
  definition~\eqref{eqn:convolutions} of $f\star g$, adopted here
  because of its convenience in the formulation of elastic potentials.
\end{remark}

The (horizontal) convolutional form taken by the relevant integral
operators, reflecting aforementioned translational invariance, prompts
the use of the two-dimensional partial Fourier transform, defined by
\begin{equation}
  \label{eqn:fourier}\widehat{f}(\tens{q}, y_3) = \Fcal[f](\tens{q}, y_3)
  := \int_{\Rbb^2}{f(y_1, y_2, y_3) e^{-\rmi(q_1y_1 + q_2y_2)} \dy_1\dy_2}
  = \int_{\Rbb^2}{f(\bfyt,y_3) e^{-\rmi\tens{q}\cdot\bfyt}} \dyt
\end{equation}
for functions of the variable $\bfyt$ that are in $L^1(\Rbb^2)$. The
coordinate space $(\tens{q},y_3)\in\Rbb^2\times\Rbb^+$ underpinning
the above Fourier transform will be referred to as the partial Fourier
space, with $\tens{q} = (q_1, q_2)$. Thanks to the celebrated Fourier
convolution theorem, convolutions become mere multiplications upon
application of the Fourier transform~\eqref{eqn:fourier}. In
particular, for integral operators of the generic
form~\eqref{generique:N:conv}, we have
\begin{lemma}\label{lem:continuous_convolution}
  Under the partial Fourier transform defined by
  equation~\eqref{eqn:fourier} and its extension to distributions, the
  partial convolution product~\eqref{eqn:convolutions} obeys the
  identity
  \begin{equation}
    \Fcal\big[f\star g\big](\tens{q})
    = \widehat{f}(-\tens{q})\widehat{g}(\tens{q}), \label{conv:theorem}
  \end{equation}
  for any pair $(f,g)$ that is convolvable in the sense of condition
  (a), (b) or (c) of Remark~\ref{convolvable}. Consequently, subject
  to the same type of convolvability conditions, integral operators of
  the form~\eqref{generique:N:conv} have the expression
  \begin{equation}
    \widehat{\Ncal[\tens{w}]}(\tens{q},x_3)
    = \int_{0}^{\infty} \ftens{H}(-\tens{q},x_3,y_3):\ftens{w}(\tens{q},y_3) \dy_3
    \label{N:conv:Fourier}
  \end{equation}
  The integral in~\eqref{eqn:fourier} is well defined in the classical
  sense only in case (a), and the distributional extension of the
  Fourier transform must be used instead for cases (b), (c), see
  e.g. the appendix ``Distributions''
  in~\cite{dautray_mathematical_2000}.
\end{lemma}
This lemma makes systematic use of the Fourier transform the
cornerstone of this work, as it will result in greatly speeding up
computations.

Our emphasis on convolvability restrictions
underpinning~\eqref{eqn:convolutions} and~\eqref{N:conv:Fourier} stems
from the fact that the kernels $\tens{H}(\bfzt,x_3,y_3)$ of interest
in this work, given by elastostatic fundamental solutions, are not in
$L^1(\Rbb^2)$ as functions of $\bfzt$ due to insufficiently fast decay
at infinity; moreover, the kernel of $\tens{\nabla}\Ncal$ also has a
non-integrable singularity at $\bfzt=\tens{0}$, and hence is not a
locally integrable function. As a result, the convolution
integral~\eqref{eqn:convolutions} applied to
$(\tens{H}\star\tens{w})(\bfxt)$ is not always defined in the
classical sense, even for ``nice'' (smooth, compactly-supported)
densities $\tens{w}$, while the properties of elastostatic fundamental
solutions make it well defined as a convolution between distributions
for any $\tens{w}$ having compact (i.e. spatially bounded) support,
ensuring the validity of Lemma~\ref{lem:continuous_convolution}. The
same features of $\tens{H}$ cause its partial Fourier transform to be
well defined as a distribution but not in the classical sense of the
integral in~\eqref{eqn:fourier}. Finally, the very fact that
fundamental solutions are fields created by singular loads (point
forces) entails treating them, and their governing equations, in the
sense of distributions. Summarizing, the framework of distribution
theory is necessary for the main components of our proposed treatment
to have clear meaning and validity.

Finally, our aim is to formulate and solve elastic-plastic contact
problems where all fields have horizontal periodicity, whereas
integral operators appearing in~\eqref{eqn:operators} are \emph{a
  priori} defined in a non-periodic setting. Indeed, the meaning of
the relevant convolutions becomes \emph{a priori} unclear for periodic
(hence not compactly supported) densities. However, the evaluation at
suitable discrete values of $\tens{q}$ of Fourier transforms of
non-periodic convolutions will be found (see
Theorem~\ref{thm:periodic_convolution}) to provide the required
relations between a periodic displacement $\tens{u}$ and the periodic
density $\tens{w}$ from which $\tens{u}$ emerges.

\section{Displacement and displacement
  gradient}\label{sec:fundamental}

In this section, we discuss the different integral operators used and
their partial Fourier representation.

\subsection{Fundamental problems}

The fundamental solutions used in this work are displacement fields
created by singular point sources in the half-space $\body$ endowed
with elastic properties. They are functions
$\tens{u}(\tens{x}, \tens{y})$ of two spatial variables, namely the
source variable $\tens{x}$ (i.e the location of an applied point
force) and the field variable $\tens{y}$. Such fundamental
displacements satisfy problems posed in terms of the Navier
elastostatic operator $\navier$, defined for a generic displacement
$\tens{v}$ by
\begin{equation}
  \navier[\tens{v}] := -\div\left(\hooke:\tens{\epsilon}[\tens{v}]\right)
\end{equation}
(differential operators being understood in this work as acting on the
field variable $\tens{y}$ unless specified otherwise). Fundamental
solutions decay at infinity, i.e. verify
$\tens{u}(\tens{x}, \tens{y})\to\tens{0}$ as
$\|\tens{y}-\tens{x}\|\to\infty$ for given $\tens{x}$.


The integral operators $\Mcal$ and $\Ncal$ of~\eqref{eqn:operators}
used in this work are defined in terms of the Mindlin fundamental
displacement $\tens{G}(\tens{x}, \bullet)$ (created in $\body$ by a
point force applied at $\tens{x}\in\body$), which moreover satisfies
the traction-free condition on $\partial\body$. In addition, the
Kelvin fundamental displacement $\tens{U}_\infty(\tens{x}, \bullet)$
(created in $\Rbb^3$ by a point force applied at $\tens{x}\in\Rbb^3$)
and the Boussinesq-Cerruti fundamental displacement
$\tens{B}(\tens{x},\bullet)$ (created in $\body$ by a point force
applied at $\tens{x}\in\partial\body$) play a key role in establishing
the expression of $\tens{G}$ that serves our purposes, so are
introduced first.

We now present in more detail the Kelvin, Boussinesq-Cerruti and
Mindlin fundamental problems. The convolutions using those kernels
will be computed in the partial Fourier space using
Lemma~\ref{lem:continuous_convolution}.

\subsubsection{The Kelvin problem}

The displacement caused by a point force in an infinite medium is at
the heart of integral equation methods in solid mechanics. Here, the
Kelvin tensor $\tens{U}_\infty = \tens{e}_k\otimes\tens{U}_\infty^k$
is used as a stepping stone to the Mindlin fundamental solution. The
displacement $\tens{U}_\infty^k$ satisfies
\begin{equation}
  \navier[\tens{U}_\infty^k](\tens{x},\bullet) = \tens{e}_k\,\delta_\tens{x} \quad\text{in }\Rbb^3, \qquad \text{ for any $\tens{x}\in\Rbb^3,\;1\leq k\leq 3$}, \label{kelvin:def}
\end{equation}
where $\delta_\tens{x}$ is the three-dimensional Dirac distribution on
$\Rbb^3$ supported at $\tens{x}$, and decays at infinity. Due to its
singular right-hand side, equation~\eqref{kelvin:def} must be
understood in the sense of
distributions\footnote{Equation~\eqref{kelvin:def} therefore means
  that $\tens{U}_\infty^k$ must verify
  $\dual{\navier[\tens{U}_\infty^k](\tens{x},\cdot)}{\tens{\phi}}{\Rbb^3}
  = \tens{e}_k\dual{\delta_\tens{x}}{\tens{\phi}}{\Rbb^3}$,
  i.e.
  $\dual{\tens{U}_\infty^k(\tens{x},\cdot)}{\navier[\tens{\phi}]}{\Rbb^3}
  = \phi_k(\tens{x})$ after integrations by parts and recalling the
  self-adjointness of the Navier operator, for any test function
  $\tens{\phi} \in C_0^\infty(\Rbb^3;\Rbb^3)$ (with
  $\dual{f}{\phi}{\Rbb^3}$ denoting the duality product). Similar
  interpretations implicitly apply for the other fundamental
  solutions.}, and the same will apply implicitly to the other
fundamental solutions and their governing equations.

The Kelvin fundamental solution possesses the full-space translational
symmetry
$\tens{U}_\infty(\tens{x}, \tens{y}) = \tens{U}_\infty(\tens{0},
\tens{y}-\tens{x})$ as well as the implied property
$\tens{\nabla}_\tens{x}\tens{U}_\infty(\tens{x}, \tens{y}) =
-\tens{\nabla}\tens{U}_\infty(\tens{x}, \tens{y})$ for its
gradients. We now define the operator $\Ncal_\infty$
\begin{equation}
  \label{eqn:n_infty}\Ncal_\infty[\tens{w}](\tens{x})
  := \int_{\body} \tens{\nabla U}_\infty(\tens{0}, \tens{y}-\tens{x}):\tens{w}(\tens{y}) \,\text{d}\tens{y}
  = \int_{0}^{\infty} \big( \tens{\nabla U}_\infty\star\tens{w} \big)(\bfxt,x_3,y_3) \dy_3.
\end{equation}
The field $\Ncal_\infty[\tens{w}]$ decays at infinity and can readily
be shown to satisfy
\begin{equation}
  \navier\big[ \Ncal_\infty[\tens{w}] \big] = -\text{div}\tens{w} \qquad \text{in $\Rbb^3$}
\end{equation}
i.e. it is the elastostatic displacement created in an unbounded
medium by an eigenstress distribution $\tens{w}$.

\subsubsection{The Boussinesq-Cerruti problem}

Of prime importance in contact
mechanics~\citep{johnson_contact_1985a}, the Boussinesq-Cerruti
problem~\citep{love_treatise_1892} gives the displacement of a
semi-infinite body $\body$ subject to a point force applied on its
surface at $\tens{x}=(\bfxt,0)$. The Boussinesq-Cerrutti tensor
$\tens{B}(\bfxt,\tens{y}) =
\tens{e}_k\otimes\tens{B}^k(\bfxt,\tens{y})$ satisfies
\begin{equation}
  \navier[\tens{B}^k](\bfxt,\bullet) = \tens{0} \quad \text{in $\body$}, \qquad
  \tens{T}[\tens{B}^k](\bfxt,\bullet) = \tilde{\delta}_{\bfxt}\tens{e}_k \quad \text{on $\partial\body$}, \qquad \text{ for any $(\bfxt, 0)\in\partial\body,\;1\leq k\leq 3$},
\end{equation}
where $\tilde{\delta}_{\bfxt}$ is the Dirac distribution defined on
$\partial\body$ and supported at $(\bfxt,0)\in\partial\body$. The
Boussinesq-Cerruti tensor will be used next for cancelling the
traction vector of the Mindlin fundamental solution.

\subsubsection{The Mindlin problem}\label{mindlin}

The problem of a point force in a semi-infinite isotropic elastic
medium with a free surface was solved
by~\citet{mindlin_force_1936}. This fundamental solution allows to
express the displacement created in $\body$ by an eigenstress
distribution $\tens{w}$ in $\body$ and satisfying a traction-free
condition on $\partial\body$. The latter feature makes it very
attractive for contact problems since it removes the need to solve an
implicit boundary integral equation, which is a staple of conventional
boundary-element methods (\citealp{bonnet_boundary_1995} and
references therein). The Mindlin tensor
$\tens{G} = \tens{e}_k\otimes\tens{G}^k$ satisfies:
\begin{equation}
  \label{eqn:fundamental_problem}
  \navier[\tens{G}^k](\tens{x},\bullet) = \delta_{\tens{x}} \quad \text{in $\body$}, \qquad
  \tens{T}[\tens{G}^k](\tens{x},\bullet) = \tens{0} \quad \text{on $\partial\body$}.
\end{equation}
The operator $\Ncal$ defined by
\begin{equation}
  \label{N:conv}\Ncal[\tens{w}](\tens{x})
  := \int_{0}^{\infty} \big( \tens{\nabla G}\star\tens{w} \big)(\bfxt,x_3,y_3) \dy_3.
\end{equation}
then yields the displacement created in $\body$ by an eigenstress
distribution $\tens{w}$. The operator $\Mcal$, which gives the
displacement in $\body$ created by surface tractions $\tens{p}$, is
given by
\begin{equation}
  \label{M:conv}\Mcal[\tens{p}](\bfxt, x_3) = (\tens{G}(\bullet, x_3, 0)\star \tens{p})(\bfxt, x_3).
\end{equation}

It is easy to see that, by virtue of linear superposition, $\ftens{G}$
is given in terms of $\ftens{U}_{\infty}$ and $\ftens{B}$ by
\begin{equation}
  \label{eqn:mindlin_B_U}\ftens{G}^k(\tens{q}, x_3, y_3)=\ftens{U}^k_{\infty}(\tens{q}, x_3, y_3)-\ftens{B}^T(\tens{q}, x_3) \cdot \ftens{T}[\ftens{U}^k_{\infty}(\tens{q}, x_3, y_3)],
\end{equation}
with the second term in~\eqref{eqn:mindlin_B_U} canceling the traction
vector on $\partial\body$. Consequently, the operators $\Mcal$ and
$\Ncal$ can be readily evaluated in the partial Fourier space once
$\ftens{B}$ and $\ftens{U}_\infty$ are known.

\subsection{Partial Fourier space
  solutions}\label{sec:fourier_solutions}

The main novelty of this work is the derivation and use of fundamental
solutions directly in the partial Fourier space. In addition to
providing substantial memory and computational savings, knowing
closed-form expressions of these fundamental solutions enables
optimizations which were previously tedious. The fundamental solutions
are found in this context by solving equations involving the
transformed Navier operator for an isotropic medium, expressed as:
\begin{equation}
  \widehat{\navier[\tens{u}]}(\tens{q},y_3)
  := \fnavier[\ftens{u}](\tens{q},y_3)
  = \mu\left\{\left(\frac{\d^2}{\d y_3^2} - q^2\right)\ftens{u}(\tens{q},y_3) + c\ftens{\nabla}\big(\ftens{\nabla}\cdot\ftens{u}(\tens{q},y_3)\big)\right\},
\end{equation}
with $c := (\lambda + \mu) / \mu = 1/(1-2\nu)$,
$\ftens{\nabla} := (\rmi q_1, \rmi q_2, \d/\d y_3)$ and
$q^2:=\|\tens{q}\|^2=q_1^2+q_2^2$. We now present our general
methodology for obtaining the desired partial-Fourier expressions of
fundamental solutions. The details of this treatment, including the
source code deriving the solutions, are available in the companion
notebook~\citep{frerot_mindlin_2018}.

\subsubsection{Basis of \texorpdfstring{$\ker(\ftens{\mathrm{N}})$}{ker(\textbf{N})}}
The process of deriving solutions for the Kelvin and
Boussinesq-Cerruti fundamental problems involves finding elements of
$\ker(\fnavier)$, the 6-dimensional space of functions $\ftens{u}$
satisfying the ODE $\fnavier[\ftens{u}]=\tens{0}$, that fulfill
specific (e.g.\ boundary or decay) conditions. We use in this work the
basis of $\ker(\fnavier)$ derived by \citet{amba-rao_fourier_1969},
defined as follows. Let the matrix-valued functions $\tens{A}^+$ and
$\tens{A}^-$ be given by
\begin{equation}
  \tens{A}^\pm(\tens{q}, y_3)
  = e^{\mp q y_3}\left(\Id + \frac{c}{c+2}q y_3\tens{\Delta}^\pm\otimes\tens{\Delta}^\pm\right),
  \label{eqn:basis_tensor}
\end{equation}
where $\tens{\Delta}^{\pm}$ and $\tens{\Delta}$ are defined in
Table~\ref{tab:full_space_solution_symbols}. Each column
$\tens{A}_j^\pm$ of $\tens{A}^\pm$ solves
$\fnavier[\tens{A}_j^\pm]=\tens{0}$, so that we have
\begin{equation}
  \label{eqn:ker(N)}
  \ker(\fnavier) = \text{range}\left(\tens{A}^+\right) + \text{range}\left(\tens{A}^-\right).
\end{equation}
Moreover, the matrices $\tens{A}^\pm$ are invertible and linearly
independent (see~\ref{app:invertibility}): any element of
$\ker(\fnavier)$ therefore has six free coefficients, to be determined
by additional conditions.

\subsubsection{The Kelvin solution}

The Kelvin problem in partial Fourier space consists in solving the
distributional ODE:
\begin{equation}
  \label{eqn:fourier_kelvin_problem}
  \fnavier[\ftens{U}_\infty^k](\tens{q}, x_3, \bullet)= \tens{e}_k\delta_{x_3}, \qquad \text{for all\ } (\tens{q}, x_3) \in \Rbb^2\times \Rbb,
\end{equation}
where $\delta_{x_3}$ is the one-dimensional Dirac distribution
supported at $x_3$. To find $\ftens{U}_\infty^k$, we follow the
methodology of~\citet{chaillat_new_2014} and seek the displacement
vector separately in each semi-infinite interval extending from the
source point:
\begin{equation}
  \ftens{U}_\infty^{k}(\tens{q}, x_3, y_3) = \begin{cases}
    \ftens{U}_\infty^{k,-}(\tens{q}, x_3, y_3) & y_3 \in\ \rbrack -\!\infty, x_3\lbrack,\\
    \ftens{U}_\infty^{k,+}(\tens{q}, x_3, y_3) & y_3 \in \lbrack x_3, +\infty\lbrack.\\
  \end{cases}\label{kelvin:function}
\end{equation}
We can, without loss of generality, set $x_3 = 0$ since
$\ftens{U}_\infty$ is invariant by translation along $\tens{e}_3$ (we
then have
$\ftens{U}_\infty^{k}(\tens{q}, x_3,
y_3)=\ftens{U}_\infty^{k}(\tens{q}, 0, y_3-x_3)$). Each contribution
$\ftens{U}_\infty^{k,\pm}$ satisfies the homogeneous Navier equation,
and hence belongs to $\ker(\fnavier)$. Using~\eqref{eqn:basis_tensor},
\eqref{eqn:ker(N)} with the requirement that
$\ftens{U}_\infty^{k,\pm}$ decay as $y_3\to\pm\infty$, we obtain:
\begin{equation}
  \label{eqn:ansatz_+}\ftens{U}_\infty^{k,\pm}(\tens{q}, 0, y_3)  = \tens{A}^\pm(\tens{q}, y_3)\cdot\tens{C}^{k,\pm},
\end{equation}
where $\tens{C}^{k,\pm}\in\Cbb^3$ are the remaining free
coefficients. The latter are determined by requiring that
$\ftens{U}_\infty^{k}$, expressed as a distribution by the single
formula (employing the Heaviside function instead
of~\eqref{kelvin:function})
\begin{equation}
  \ftens{U}_\infty^{k}(\tens{q}, x_3, y_3)
  = H(y_3-x_3)\ftens{U}_\infty^{k,+}(\tens{q}, x_3, y_3)
  + \big( 1-H(y_3-x_3) \big) \ftens{U}_\infty^{k,-}(\tens{q}, x_3, y_3) \label{Uinfty:single}
\end{equation}
should be continuous at $y_3=x_3$ and should satisfy the
distributional ODE~\eqref{eqn:fourier_kelvin_problem} (recalling that
$H'(\bullet-x_3)=\delta_{x_3}$). As can be seen in the companion
notebook, this results in the following expression for
$\ftens{U}_\infty^\pm = \tens{e}_k \otimes \ftens{U}_\infty^{k,\pm}$:
\begin{equation}
  \label{eqn:full_space_solution}\ftens{U}_\infty^{\pm}(\tens{q}, x_3, y_3)  = \frac{1}{q}\left[\ftens{U}_{0,0}^{\pm}(\tens{q}) g_0^\pm\big(q(y_3-x_3)\big) + \ftens{U}_{1,0}^{\pm}(\tens{q}) g_1^\pm\big(q(y_3-x_3)\big)\right].
\end{equation}
The symbols of equation~\eqref{eqn:full_space_solution} are defined in
Table~\ref{tab:full_space_solution_symbols}. On observing that the
functions $g_0^\pm(z)$ and $g_1^\pm(z)$ verify
\begin{equation}
  \ftens{\nabla}g_0^\pm(qy_3)
  = \mp q\tens{\Delta}^\pm g_0^\pm(qy_3),\qquad
  \ftens{\nabla}g_1^\pm(qy_3)
  = \mp q\tens{\Delta}^\pm g_1^\pm(qy_3) + q\tens{e}_3 g_0^\pm(qy_3),
\end{equation}
the regular parts of the distributional derivatives of
$\ftens{U}_\infty$ at any order are easily found to be given through
the recurrence relations
\begin{equation}\label{eqn:recurrence}
  \begin{aligned}
    \widehat{\tens{\nabla}^n \tens{U}_\infty^{\pm}}(\tens{q}, x_3, y_3) & = q^{n-1}\left[\ftens{U}_{0,n}^{\pm}(\tens{q}) g_0^\pm\big(q(y_3-x_3)\big) + \ftens{U}_{1,n}^{\pm}(\tens{q}) g_1^\pm\big(q(y_3-x_3)\big)\right],\\
    \ftens{U}_{0,n}^\pm(\tens{q}) & = \mp \ftens{U}_{0,n-1}^\pm\otimes\tens{\Delta}^\pm + \ftens{U}_{1,n-1}^\pm\otimes\tens{e}_3,\\
    \ftens{U}_{1,n}^\pm(\tens{q}) & = \mp
    \ftens{U}_{1,n-1}^\pm\otimes\tens{\Delta}^\pm.
  \end{aligned}
\end{equation}
The Kelvin tensor appears in the operators $\Ncal_{\infty}$ and
$\Ncal$ only through its first-order gradient, see~\eqref{eqn:n_infty}
and~\eqref{eqn:mindlin_B_U}. Since $\ftens{U}_\infty$ is continuous at
$y_3=x_3$, the representation~\eqref{Uinfty:single} shows that
$\ftens{\nabla U}_\infty$ can be identified with a discontinuous
function (in the classical sense) and evaluated on the sole basis of
formulas~\eqref{eqn:recurrence}.

\begin{table}
  \centering
  \caption{Symbols for the full-space fundamental
    solution}\label{tab:full_space_solution_symbols}
  \begin{tabular}{r l}
    \toprule
    Symbol & Expression\\
    \midrule
    $b$ & $4(1-\nu)$\\
    $\tens{\Delta}$, $\tens{\Delta}^\pm$ & $(\rmi q_1/q, \rmi q_2/q, 0)$, $\tens{e}_3\mp\tens{\Delta}$ \\
    \midrule
    $\ftens{U}_{0,0}^-$ & $\dfrac{1}{2\mu b}\left(b\Id + \tens{\Delta}\otimes\tens{\Delta} - \tens{e}_3\otimes\tens{e}_3\right)$\\[0.3cm]
    $\ftens{U}_{0,0}^+$ & $\dfrac{1}{2\mu b}\left(b\Id + \tens{\Delta}\otimes\tens{\Delta} - \tens{e}_3\otimes\tens{e}_3\right)$\\
    \midrule
    $\ftens{U}_{1,0}^-$ & $-\dfrac{1}{2\mu b} \tens{\Delta}^-\otimes\tens{\Delta}^-$\\[0.3cm]
    $\ftens{U}_{1,0}^+$ & $\dfrac{1}{2\mu b} \tens{\Delta}^+\otimes\tens{\Delta}^+$\\
    \midrule
    $g_0^\pm(z)$, $g_1^\pm(z) $ & $e^{\mp z}$, $ze^{\mp z}$ \\
    \bottomrule
  \end{tabular}
\end{table}

Then, use of Lemma~\ref{lem:continuous_convolution} allows to express
the operator $\Ncal_{\infty}$ in Fourier space as
\begin{equation}
  \widehat{\Ncal_{\infty}[\tens{w}]}(\tens{q},x_3)
  = \int_{0}^{\infty} \ftens{\nabla U}_{\infty}(-\tens{q},y_3 - x_3):\ftens{w}(\tens{q},y_3) \dy_3 \label{eqn:fourier_eigenstress}
\end{equation}

\subsubsection{Boussinesq-Cerruti solution}

The Boussinesq-Cerruti fundamental tensor
$\ftens{B} = \tens{e}_k \otimes \ftens{B}^k$ satisfies:
\begin{equation}
  \fnavier[\ftens{B}^k](\tens{q}, y_3) = \tens{0}\quad\text{for all } (\tens{q}, y_3)\in\Rbb^2\times\Rbb^+,\qquad
  \tens{T}[\ftens{B}^k](\tens{q}) = \tens{e}_k.
\end{equation}
Using only the basis tensor $\tens{A}^+$ (thereby enforcing the decay
condition as $y_3\rightarrow +\infty$) and identifying the free
coefficients, we find that $\ftens{B}$ is given by
\begin{equation}
  \ftens{B}(\tens{q}, y_3)
  = \frac{1}{q}\left[\ftens{B}_{0,0}(\tens{q})g_0^+(q y_3) + \ftens{B}_{1,0}(\tens{q})g_1^+(q y_3)\right],
  \label{B:fourier}
\end{equation}
with
\begin{align}
  \ftens{B}_{0,0}(\tens{q}) & = \frac{1}{2\mu}\left(2\Id + (1-2\nu)\tens{\Delta}^+\otimes\tens{\Delta}^- + \tens{\Delta}\otimes\tens{\Delta} - \tens{e}_3\otimes\tens{e}_3\right),\\
  \ftens{B}_{1,0}(\tens{q}) & = \frac{1}{2\mu}\tens{\Delta}^+\otimes\tens{\Delta}^+.\label{eqn:B:1_0}
\end{align}
Then, the recurrence relations~\eqref{eqn:recurrence} are also valid
for the gradients of $\ftens{B}$. We can now readily construct
$\ftens{G}$ and its gradient $\ftens{\nabla G}$ using
equation~\eqref{eqn:mindlin_B_U} and the recurrence relations for the
regular parts of $\ftens{U}_\infty$ and $\ftens{B}$.

\subsection{Displacement gradient
  computation}\label{sec:displacement_gradient}

Due to the construction of the Mindlin fundamental solution in
equation~\eqref{eqn:mindlin_B_U}, the evaluation of $\tens{\nabla u}$
requires the computation of $\tens{\nabla}\Ncal_\infty$. This operator
is singular (see
e.g.~\citep{bui_remarks_1978,bonnet_modified_2017,gintides_solvability_2015}),
but the present distributional and partial-Fourier framework still
allows for a very straightforward treatment. Indeed, we simply have
\begin{align}
  \widehat{\tens{\nabla}\Ncal_\infty[\tens{w}]}(\tens{q}, x_3)
  & = - \Big( \widehat{\tens{\nabla}^2\tens{U}_\infty}(-\tens{q}, \bullet)\star\ftens{w} \Big) \nonumber\\
  & = - \int_{0}^{x_3}q
    \left\{
    g_0^{-}\big(q(y_3 - x_3)\big)
    \left[
    \big(\ftens{U}_{0,1}^-(-\tens{q}):\ftens{w}(\tens{q},y_3)\big) \otimes \tens{\Delta}^- + \big(\ftens{U}_{1,1}^-(-\tens{q}):\ftens{w}(\tens{q},y_3)\big) \otimes \tens{e}_3
    \right]\right.\nonumber\\
  & \quad\quad\quad\quad\:{} + g_1^-\big(q(y_3-x_3)\big)
    \left.\big(
    \ftens{U}_{1, 1}^-(-\tens{q}):\ftens{w}(\tens{q}, y_3)\big) \otimes \tens{\Delta}^-
    \right\}\dy_3\nonumber\\
  & \quad{} - \int_{x_3}^{\infty}q
    \left\{
    g_0^{+}\big(q(y_3 - x_3)\big)
    \left[
    - \big(\ftens{U}_{0, 1}^+(-\tens{q}):\ftens{w}(\tens{q}, y_3)\big) \otimes \tens{\Delta}^+ + \big(\ftens{U}_{1, 1}^+(-\tens{q}):\ftens{w}(\tens{q}, y_3)\big) \otimes \tens{e}_3
    \right]\right.\nonumber\\
  & \quad\quad\quad\quad\:{} - g_1^+\big(q(y_3-x_3)\big)
    \left.
    \big(\ftens{U}_{1,1}^+(-\tens{q}):\ftens{w}(\tens{q},y_3)\big) \otimes \tens{\Delta}^+
    \right\}\dy_3\nonumber\\
  \label{eqn:2gradient}& \quad {} - \left(\big\llbracket\widehat{\tens{\nabla}\tens{U}_\infty}\big\rrbracket(\tens{q}):\ftens{w}(\tens{q}, x_3)\right)\otimes\tens{e}_3,
\end{align}
with the second equality resulting from the application of
$\ftens{\nabla}$ to $\ftens{\nabla U}_\infty$ as given
by~\eqref{Uinfty:single} and~\eqref{eqn:recurrence}. The discontinuity
jump in the gradient of $\ftens{U}_\infty$ is computed,
using~\eqref{eqn:full_space_solution} and~\eqref{eqn:recurrence}, as:
\begin{equation}
  \big\llbracket \widehat{\tens{\nabla}\tens{U}_\infty}\big\rrbracket(\tens{q})
  = \ftens{U}_{0,1}^+(\tens{q}) - \ftens{U}_{0, 1}^-(\tens{q})
  = \frac{1}{\mu b}\left(2\tens{e}_3\otimes\tens{e}_3 - b\Id\right)\otimes\tens{e}_3, \label{jump}
\end{equation}
with the vector $\tens{e}_3$ post-multiplying the discontinuity term
in~\eqref{eqn:2gradient} coming from the fact that distributional
terms can only arise from derivatives in the variable $y_3$ using the
present partial-Fourier framework.
\begin{remark}\label{CPV}
  It is interesting to compare the formulations of
  $\tens{\nabla}\Ncal_\infty$ in the partial Fourier and physical
  spaces. In the latter case, the function
  $\tens{\nabla^2 U}_\infty(\tens{x},\tens{y})$ has a non-integrable
  $|\tens{y}-\tens{x}|^{-3}$ singularity. Consequently, the singular
  integral operator is the sum of a Cauchy principal value (CPV)
  integral and a free term, as pointed out e.g.\
  in~\citet{bui_remarks_1978} (for example, a careful distributional
  interpretation of the application of $\tens{\nabla}$ to
  $\Ncal_\infty[\tens{w}]$ yields both contributions). Then, practical
  evaluations of $\tens{\nabla}\Ncal_\infty$ in the physical space
  entail special methods for the integration of a
  CPV~\citep{guiggiani_general_1990}. By contrast, the present method,
  which implicitly and indirectly accounts for both the CPV and the
  free term contributions (the partial-Fourier counterpart of the
  latter being the jump term~\eqref{jump}) is significantly easier to
  exploit numerically, as no specialized methods are required.
\end{remark}


\section{Discretized operators}\label{sec:discretized_operators}

In this section, we study the numerical evaluation of the integral
operator $\Ncal[\tens{w}]$ for given discretized densities $\tens{w}$
(the evaluation of $\Mcal[\tens{p}]$, $\tens{\nabla}\Mcal[\tens{p}]$
and $\tens{\nabla}\Ncal[\tens{w}]$ can then be formulated similarly by
adapting the considerations made for $\Ncal[\tens{w}]$). This extends
the developments of~\citet{zeman_finite_2017} to our partial Fourier
representation. Although the use of discrete Fourier methods is
widespread in simulating the contact of rough surfaces (e.g.\
\cite{polonsky_fast_1999, jacq_development_2002, wang_new_2013} which
use DFT of real-space fundamental solutions
and~\cite{yastrebov_contact_2012, rey_normal_2017,
  weber_molecular_2018} which use Fourier-space fundamental solutions)
and dates back to~\citet{stanley_fftbased_1997} (which implicitly uses
a Fourier space fundamental solution
via~\citealp{johnson_contact_1985}), a theoretical basis for the
discretization of continuous operators has, to the best of our
knowledge, never been provided.

\subsection{Spectral discretization and
  DFT}\label{sec:spectral_interpolation}

Let $\tens{L} = (L_1, L_2, L_3) \in {\mathbb{R}^{+}}^3$ be the three
lengths defining the discretized domain and
$\tens{N} = (N_1, N_2, N_3) \in \mathbb{N}^3$ the number of points
considered in each direction. Let us also define the following sets:
\begin{align}
  \Zbb^2_{\tens{N}} & = \left\{\tens{k} \in \Zbb^2 : -\frac{N_1}{2} < k_1 < \frac{N_1}{2}, -\frac{N_2}{2} < k_2 < \frac{N_2}{2}\right\},\\
  X_3 & = {\left\{x_3^i\right\}}_{i=1}^{i=N_3}\subset\mathbb{R}^+\quad\text{with } L_3 = \sup X_3 - \inf X_3,\\
  X & = \left\{k_1\frac{L_1}{N_1}\tens{e}_1 + k_2\frac{L_2}{N_2}\tens{e}_2 + x_3 \tens{e}_3 : \tens{k}\in\Zbb^2_{\tens{N}}\ \text{ and }\ x_3 \in X_3\right\}\subset \Bcal_p.
\end{align}
We recall from~\eqref{eqn:periodic_cell} that $\Bcal_p$ is a
semi-infinite cell of size $L_1\times L_2$ when projected on
$\partial\body$. $X$ is a set of $N_1\times N_2\times N_3$ discrete
points, which projected on the $(Ox_1x_2)$ plane forms a regular grid,
while the projection on $(Ox_3)$ gives the set of chosen positive
values $X_3$. Unlike full-space Fourier
methods~\citep{moulinec_numerical_1998, wang_new_2013,
  zeman_finite_2017}, we are free to choose the spacing of points in
$X_3$. Any $\Bcal_p-$periodic eigenstress $\tens{w}$, as well as the
$\Bcal_p-$periodic displacement $\tens{u}=\Ncal[\tens{w}]$, can be
expressed as complex Fourier series:
\begin{subequations}
  \begin{align}
    \tens{w}(\bfxt, x_3)
    = \sum_{\tens{k}\in \Zbb^2} \ftens{\mathrm{w}}(\tens{k}, x_3)\exp(2\pi \rmi\bar{\tens{k}}\cdot\bfxt), \label{w:fourier} \\
    \tens{u}(\bfxt, x_3)
    = \sum_{\tens{k}\in\Zbb^2} \ftens{\mathrm{u}}(\tens{k},x_3)\exp(2\pi \rmi\bar{\tens{k}}\cdot\bfxt),
    \label{u:fourier}
  \end{align}
\end{subequations}
where $\ftens{\mathrm{w}}(\tens{k}, x_3)$ and
$\ftens{\mathrm{u}}(\tens{k}, x_3)$ are the Fourier coefficients of
the series and $\bar{k}_i = k_i / L_i$. Because $\tens{w}$ is
$\body_p$-periodic, it is no longer convolvable with
$\tens{\nabla G}$, making it impossible to evaluate $\Ncal[\tens{w}]$
as a convolution (the same remark applies to
$\Ncal_{\infty}[\tens{w}]$ and the other operators). The operator
$\Ncal[\tens{w}]$ can nevertheless still be evaluated by means of the
non-periodic partial Fourier representation of the fundamental
solution obtained in Section~\ref{sec:fourier_solutions}, thanks to
the following result:
\begin{theorem}\label{thm:periodic_convolution}
  Let $\tens{w}$ be $\Bcal_p$-periodic. Then $\Ncal[\tens{w}]$ is
  $\Bcal_p$-periodic and
  \begin{equation}\label{eqn:discrete_convolution}
    \Ncal[\tens{w}](\bfxt,x_3)
    = \frac{1}{4\pi^2} \sum_{\tens{k}\in\Zbb^2} \Big( \int_0^{\infty}
    \ftens{\nabla G}(-2\pi \bar{\tens{k}},x_3,y_3):\ftens{\mathrm{w}}(\tens{k},y_3) \dy_3 \Big)\exp(2\pi \rmi \bar{\tens{k}}\cdot\bfxt).
  \end{equation}
\end{theorem}
\begin{proof}
  See~\ref{nonper:to:per}.
\end{proof}

\begin{remark}\label{rem:fundamental_mode}
  All the fundamental solutions presented in
  Section~\ref{sec:fourier_solutions} have a $O(q^{-1})$ weak
  singularity at $\tens{q}=\tens{0}$, which does not prevent normal
  use of their continuous inverse Fourier transforms. By contrast,
  discrete transform evaluation at $\bar{\tens{k}} = \tens{0}$ is not
  possible. For computing displacements, we can arbitrarily set
  e.g. $\ftens{G}(\tens{0},\bullet)=\tens{0}$, following common
  practice~\citep{stanley_fftbased_1997, zeman_finite_2017}. In this
  work, only the operator $\tens{p}\mapsto\Mcal[\tens{p}]$ requires
  this adjustment, as all other operators involve gradients of
  fundamental solutions, which have no singularity in $\tens{q}$,
  see~\eqref{eqn:recurrence}. The displacement resulting from the
  ``regularized'' operator is correct up to an additive constant and a
  linear displacement depending only on $x_3$. The former will be
  determined by contact conditions (i.e.\ imposed average load or
  displacement), while the latter can be ignored: it diverges in a
  semi-infinite domain, and we are only interested in surface
  displacements for the contact problem. Note that it produces a
  gradient constant w.r.t $x_3$, which is taken into account in
  $\tens{\nabla\Ncal}$. We suppose however that this gradient alone is
  not enough to trigger a plastic response of the material.
\end{remark}

Unlike the method of~\citet{sainsot_numerical_2002} (and subsequent
adaptations by \citealp{chen_fast_2008} and \citealp{wang_new_2013}),
the direct use of the closed-form expression of $\ftens{\nabla G}$
implies that equation~\eqref{eqn:discrete_convolution} is exact. The
discretization error inherent in the numerical calculation of
$\ftens{\nabla G}$ based on the discrete Fourier transform of the
real-space fundamental solution
$\tens{\nabla G}$~\citep{firth_discrete_1992, boyd_chebyshev_2001} is
therefore avoided, which is a definitive advantage over the
previously-mentioned sampling methods. Moreover, storage of the
discrete values $\ftens{\nabla G}(-2\pi\bar{\tens{k}},x_3,y_3)$ is not
necessary, yielding substantial memory gains, especially for higher
order operators such as $\tens{\nabla}\Ncal$ involving the
fourth-order tensor $\tens{\nabla}^2\tens{G}$.

Implementing the proposed method nevertheless entails unavoidable
approximations. One stems from the necessary truncation of the Fourier
series in~\eqref{eqn:discrete_convolution}. Another appears when
$\tens{w}$ results from plastic deformations, which require a local
(i.e.\ physical space) representation in order to perform operations
such as the return mapping procedure. Hence, the Fourier coefficients
$\tens{\mathrm{w}}$ are approximated using the discrete Fourier
transform~\citep{firth_discrete_1992, zeman_finite_2017}
\begin{equation}\label{eqn:dft}
  \ftens{\mathrm{w}}_h := \mathrm{DFT}[\tens{w}\big|_X]
\end{equation}
which is known to cause discretization
errors~\citep{boyd_chebyshev_2001}. Note that because operations like
the computation of plastic deformations are intrinsically local in the
physical space and the application of integral operators (e.g.\
$\Ncal$) is local in the partial-Fourier space (for the $x_1$ and
$x_2$ directions), the solving of the elastic-plastic problem will
involve going back and forth between the physical and the
partial-Fourier representations, using the discrete Fourier
transform. Consequently, numerical evaluation of
equation~\eqref{eqn:dft} is done with the FFT
algorithm~\citep{cooley_algorithm_1965} because of its very attractive
$\Ocal(N_1N_2\log(N_1N_2))$ computational complexity, by means of the
open-source library FFTW~\citep{frigo_design_2005} for the present
implementation.

\subsection{Discretization and integration in the $x_3$
  direction}\label{sec:x_3_discrete}

Equation~\eqref{eqn:discrete_convolution} involves integrals in the
$x_3$ direction. The purpose of this section is to present the
procedure developed to compute them. First, we introduce a generic
interpolation of $\ftens{\mathrm{w}}$ in the $x_3$ direction:
\begin{equation}
  \ftens{\mathrm{w}}_h(\tens{k}, x_3) = \sum_{j = 1}^{N_3}{\ftens{\mathrm{w}}_j(\tens{k})\phi_j(x_3)},
\end{equation}
where $\phi_j$ (resp. $\ftens{\mathrm{w}}_j$) are the interpolation
function (resp. the Fourier coefficients of $\tens{w}$) evaluated at
$x_3\in X_3$.  From equation (\ref{eqn:discrete_convolution}), the
evaluation of $\Ncal[\tens{w}]$ takes the form:
\begin{equation}\label{eqn:full_discrete_convolution}
  \Ncal[\tens{w}](\bfxt,x_3)
  = \frac{1}{4\pi^2} \sum_{\tens{k}\in\Zbb^2} \sum_{j=1}^{N_3} \Big( \int_0^{\infty} \ftens{\nabla G}(-2\pi \bar{\tens{k}},x_3,y_3)\phi_j(y_3) \dy_3 \Big): \ftens{\mathrm{w}}_j(\tens{k})\,\exp(2\pi \rmi \bar{\tens{k}}\cdot\bfxt).
\end{equation}
Considering a specific node $x_3^i \in X_3$, the associated Fourier
coefficients in equation~\eqref{eqn:full_discrete_convolution} are
computed as a weighted sum of $N_3$ convolution
integrals. Equation~\eqref{eqn:full_space_solution} reveals, after
re-arrangements, that the integral
in~\eqref{eqn:full_discrete_convolution} can be expressed in terms of
integrals of the simpler form
\begin{equation}\label{eqn:typical_integral}
  \int_0^\infty{g_k^\pm\big(q(y_3 - x_3^i)\big)\phi_j(y_3)\,\d y_3}.
\end{equation}
with $k = 0, 1$. Furthermore, $\Ncal[\tens{w}]$ has to be evaluated at
every node $x_3^i \in X_3$, so that in the worst case the total number
of integrals (\ref{eqn:typical_integral}) to be computed is
$\Ocal(N_3^2)$.

This cost may however be mitigated, and we now propose a method whose
efficiency is better than that of the naive approach consisting in
evaluating equation~\eqref{eqn:typical_integral} for all
$1\leq i,j\leq N_3$. We choose classical Lagrange polynomials as our
basis of interpolation functions. Let $E_i = [x_3^i, x_3^{i+1}]$
($i \in \{1,\ldots,N-1\}$) be an element with Lagrange polynomials
$\phi_j^L$ ($j \in \{1,\ldots,n\}$) of degree $n-1$; the center
$x_c^i$ and half-length $e_i$ of $E_i$ are given by
$x_c^i=\tfrac{1}{2}(x^i+x^{i+1})$ and $e_i=x_c^i-x^i$. Using the
change of variables $y=x^i_c+ze_i$ and the properties of exponential
functions, we easily find
\begin{subequations}
  \begin{align}
    \int_{E_i}{g_0^\pm\big(q(y_3-x_3)\big)\phi_j^L(y_3)\,\d y_3} & = g_0^\pm\big(q(x_c^i -x_3)\big)\,G_0^\pm(q, i, j),\label{eqn:int_g0}\\
    \int_{E_i}{g_1^\pm\big(q(y_3-x_3)\big)\phi_j^L(y_3)\,\d y_3} & = g_0^\pm\big(q(x_c^i -x_3)\big)\,\Big\{G_1^\pm(q, i, j) + q(x_c^i -x_3)G_0^\pm(q, i, j)\Big\},\label{eqn:int_g1}
    \intertext{with $G_k^\pm(q, i, j)$ given using the standard Lagrange polynomials $\bar{\phi}_j^L$ of degree $n-1$ defined on $[-1, 1]$ by}
    G_k^\pm(q, i, j) &:= e_i \int_{-1}^1 {g_k^\pm(qz e_i)\bar{\phi}_j^L(z)\,\dz}.
  \end{align}
\end{subequations}

It follows from~\eqref{eqn:int_g0} and~\eqref{eqn:int_g1} that the
inner integral in equation~\eqref{eqn:full_discrete_convolution}
decays as $\exp(-q|x_c^i - x_3|)$, allowing us to define a threshold
criterion $q|x_c^i - x_3| < \epsilon_{co}$ for deciding which
integrals need to be computed for the application of $\Ncal$. In
practice, if $|x^i_c-x_3|>\epsilon_{\text{co}}/q$, the integral over
$E_i$ is not computed. This results in a drastic reduction of the
number of operations needed to compute the integral
of~\eqref{eqn:discrete_convolution}, especially for large wave vectors
$\tens{k}$ since the cutoff length $\epsilon_{\text{co}}/q$ is
inversely proportional to $q = 2\pi\|\tens{k}\|$ (see~\ref{sec:app-complexity}).

In addition, while the support of
$x_3\mapsto\ftens{\mathrm{w}}_h(\tens{k},x_3)$ is \emph{a priori}
unknown in an elastic-plastic simulation, it is possible that just a
few points in a given sub-surface layer show non-zero values of
$\ftens{\mathrm{w}}_h$, in which case $\ftens{\mathrm{w}}_i$ can be
treated as a sparse vector, keeping track of the non-zero entries needed for evaluating integrals.\enlargethispage*{10ex}



\section{Elastic-plastic integral equation method}\label{sec:implicit}

The use of domain integral equation methods for elastic-plastic
analysis is now well established~\citep{telles_application_1979,
  telles_implicit_1991, bonnet_implicit_1996, gao_effective_2000,
  yu_development_2010}. In this work, we use the implicit integral
equation formulation proposed by~\citet{telles_implicit_1991}. with
yield function
$f\ysub : \mathbb{R}^{3\times3}_\text{sym} \rightarrow \mathbb{R}$ and
hardening function $f\hsub : \mathbb{R}\rightarrow\mathbb{R}$,

\subsection{Von Mises plasticity}\label{sec:von_mises}

We limit, without loss of generality, the results of this paper to von
Mises plasticity, for which the yield function
$f\ysub : \mathbb{R}^{3\times3}_\text{sym} \rightarrow \mathbb{R}$ is
given by
\begin{equation}
  f\ysub(\tens{\sigma}) := \sqrt{\tfrac{3}{2}}\|\tens{s}\|\quad\text{where }\tens{s} := \tens{\sigma}-\tfrac{1}{3}\Tr(\tens{\sigma})\tens{I}.
\end{equation}
We define the cumulated equivalent plastic strain as:
\begin{equation}
  e_p := \sqrt{\frac{2}{3}} \int_{t_0}^t{\|\dot{\tens{\epsilon}}^p\|\,\d t},
\end{equation}
where $\dot{\tens{\epsilon}}^p$ is the plastic strain rate. The
plasticity conditions are then written as:
\begin{subequations}
  \begin{align}
    f\ysub(\tens{\sigma}) - f\hsub(e_p) & \leq 0,\\
    \big(f\ysub(\tens{\sigma}) - f\hsub(e_p)\big)\dot{e}_p & = 0,\\
    \intertext{where $f\hsub : \mathbb{R}\rightarrow\mathbb{R}$ is the hardening function, with the associated flow rule}
    \dot{\tens{\epsilon}}^p &= \frac{3\dot{e}_p}{2f\ysub(\tens{\sigma})}\tens{s}(\tens{\sigma}).
  \end{align}
\end{subequations}
Although multiple choices are possible for $f\hsub$, we only consider
linear isotropic hardening, for which $f\hsub$ is given in terms of
the initial yield stress $\sigma_{\text{y}}$ and the hardening modulus
$E\hsub$ by
\begin{equation}
  f\hsub(e_p) = \sigma_{\text{y}} + E\hsub e_p.
\end{equation}

\subsection{Implicit equilibrium equation}

The elastic-plastic state $S_n$ of $\body$ at step $t_n$ is
characterized by the cumulated plastic equivalent strain $e_n^p$ and
the total plastic strain $\tens{\epsilon}_n^p$; we write
$S_n = (e_n^p, \tens{\epsilon}_n^p)$. At $t_n$, the total strain in
$\body$ can be expressed as a function of the applied boundary
tractions $\tens{t}_n^D$ (known) and the plastic strain
$\tens{\epsilon}_n^p$ (unknown):
\begin{equation}
  \label{eqn:total_strain_n}\tens{\epsilon}_n = \tens{\nabla}^\mathrm{sym}\Mcal[\tens{t}_n^D] + \tens{\nabla}^\mathrm{sym}\Ncal[\hooke:\tens{\epsilon}_n^p].
\end{equation}
Writing equation~\eqref{eqn:total_strain_n} at step $t_{n+1}$ and
taking the difference of the two, we obtain the following implicit
incremental equation:
\begin{equation}
  \label{eqn:incremental_strain}\Delta\tens{\epsilon}_n = 
  \tens{\nabla}^\mathrm{sym}\Mcal[\Delta\tens{t}_n^D] + \tens{\nabla}^\mathrm{sym}\Ncal\big[\hooke:\Delta\tens{\epsilon}_n^p(\Delta\tens{\epsilon}_n;S_n) \big].
\end{equation}
We can see that this equation combines the balance of momentum
equation~\eqref{eqn:equilibrium}, the kinematic
compatibility~\eqref{eqn:epsilon} (through the use of the symmetrized
gradient $\tens{\nabla}^\mathrm{sym}$) and the constitutive behavior
through~\eqref{eqn:sigma} and
$\Delta\tens{\epsilon}_n^p(\Delta\tens{\epsilon}_n;S_n)$ (via a return
mapping algorithm symbolized by
$\Delta\tens{\epsilon}_n^p(\Delta\tens{\epsilon}_n;S_n)$). Therefore,
solving~\eqref{eqn:incremental_strain} for $\Delta\tens{\epsilon}_n$
gives the solution of the elastic-plastic problem. To achieve this, we
seek the root of the residual function:
\begin{equation}
  \Rcal[\Delta\tens{\epsilon}_n]
  := \Delta\tens{\epsilon}_n - \tens{\nabla}^\mathrm{sym}\Mcal[\Delta\tens{t}_n^D] - \tens{\nabla}^\mathrm{sym}\Ncal\big[\hooke:\Delta\tens{\epsilon}_n^p(\Delta\tens{\epsilon}_n;S_n) \big].
\end{equation}

Equation~\eqref{eqn:incremental_strain} being non-linear in
$\Delta\tens{\epsilon}_n$, we have to use an iterative method. Since
the operators $\tens{\nabla}^\mathrm{sym}\Mcal$ and
$\tens{\nabla}^\mathrm{sym}\Ncal$ applied to given arguments are very
efficiently evaluated thanks to the use of the FFT, we want to avoid
assembling them. Newton-Krylov solvers (see e.g.\ references
in~\citep{knoll_jacobianfree_2004}) are traditionally well suited for
this type of approach, especially since the consistent tangent
operator for equation~\eqref{eqn:incremental_strain} is
known~\citep{bonnet_implicit_1996}. However, we will present here a
convenient Jacobian-less method~\citep{lacruz_spectral_2006} to
solve~\eqref{eqn:incremental_strain}, as it offers several advantages
over Newton-Krylov solvers, and is readily available as part of the
SciPy~\citep{jones_scipy_2001} library.

\subsection{Jacobian-Free Spectral Residual Method}

The DF-SANE algorithm developed by~\citet{lacruz_spectral_2006} is an
attractive algorithm because of its low memory requirements compared
to traditional Krylov solvers which require storage of span vectors
for a subpart of Krylov space. It also simplifies the implementation,
as it does not require knowledge\footnote{Even approximation of the
  Jacobian via finite differences.} nor evaluation of the consistent
tangent operator. The iteration goes as:
\begin{equation}
  \delta\tens{\epsilon}_n^{i} = - \alpha_i\sigma_i\Rcal[\Delta\tens{\epsilon}_n^i], \qquad
  \Delta\tens{\epsilon}_n^{i+1} = \Delta\tens{\epsilon}_n^{i} + \delta\tens{\epsilon}_n^{i}
\end{equation}
with
\begin{equation}
  \sigma_i
  = \frac{\|\delta\tens{\epsilon}_n^{i-1}\|^2}{\langle\delta R^{i-1}, \delta\tens{\epsilon}_n^{i-1}\rangle_{L_2}}, \qquad
  \delta R^i = \Rcal[\Delta\tens{\epsilon}_n^{i+1}] - \Rcal[\Delta\tens{\epsilon}_n^i]
\end{equation}
where $\langle \bullet, \bullet \rangle_{L_2}$ is the appropriate
scalar product on $L^2(\body;\mathbb{R}^{3\times 3})$ and $\alpha_i$
is a step size determined by a line search on
$\|\Rcal[\Delta\tens{\epsilon}]\|^2$. More details can be found
in~\citep{lacruz_spectral_2006, birgin_spectral_2014}.

\section{Elastic-plastic contact}\label{sec:coupling}

In order to solve the full elastic-plastic contact problem, one needs
to solve the unknown boundary tractions $\tens{T}[\tens{u}]$ and the
plastic deformations $\tens{\epsilon}^p$. The displacement field
satisfying equilibrium once these quantities have been resolved is
obtained with equation~\eqref{eqn:operators}. The difficulty of the
elastic-plastic contact problem resides in the coupling between the
contact and the plastic problems: a change in plastic deformation will
displace the surface, changing the gap function and therefore the
contact solution, while a change in surface traction will influence
the plastic deformations.

In this work, we adopt an alternating coupling
strategy~\citep{jacq_development_2002}: the contact problem is solved
with a fixed distribution of plastic deformations, then the resulting
contact tractions are used to update the plastic deformations,
changing the residual displacements of the surface. The latter are
used to compute a new contact solution, and so on until
convergence. We first describe how the elastic contact problem is
solved, then lay out the coupling algorithm.


\subsection{Elastic contact}\label{sec:elastic_contact}
Since displacement solution of the elastic contact problem is
$\body_p$-periodic and computed using
Theorem~\ref{thm:periodic_convolution} applied to $\Mcal$ (taking into
account remark~\ref{rem:fundamental_mode}), we define the following
function space:
\begin{equation}
  \bar{H}^1(\body_p;\mathbb{R}^3) = \left\{\tens{u} + \bar{\tens{u}}\ \Big|\ \tens{u} \in H^1(\body;\mathbb{R}^3), \bar{\tens{u}} \in \mathbb{R}^3\text{ with } \ftens{\mathrm{u}}(\tens{0}, \bullet) = 0 \text{ and } \tens{u} \text{ is } \body_p\text{-periodic} \right\},
\end{equation}
which is a space of $\body_p$-periodic functions whose fundamental
mode (i.e.\ the horizontal average) is constant with respect to
$x_3$. The space of admissible solutions to the contact problem is:
\begin{equation}
  \Lambda := \left\{\tens{u} \in \bar{H}^1(\body_p;\mathbb{R}^3) : g[\tens{u}] \geq 0\right\}.
\end{equation}
The elastostatic contact problem can be written as a minimization
principle~\citep{duvaut_inequations_1972}, for a fixed plastic strain
distribution $\tens{\epsilon}_p$:
\begin{equation}
  \label{eqn:variational}\inf_{\tens{u}\in\Lambda}\left\{\frac{1}{2}\int_\body{\tens{\sigma}[\tens{u}, \tens{\epsilon}^p]:\left(\tens{\epsilon}[\tens{u}] - \tens{\epsilon}^p\right)\,\d V}\right\}.
\end{equation}
An additional constraint is needed to determine a unique solution. One
can prescribe the value of $\tens{\bar{u}}$ (the horizontal average of
the displacement) or impose the total load acting on $\partial\body$
as in equation~\eqref{eqn:total_pressure}.

In linear elastic contact, it is usual to rewrite
problem~\eqref{eqn:variational} as a minimization principle
on~$\partial\body$~\citep{kalker_variational_1977}, as it can be
easily exploited by boundary methods. This is usually done by assuming
$\tens{u}$ satisfies~\eqref{eqn:equilibrium}, which is automatically
enforced in boundary and volume integral methods. We can apply a
similar procedure via linear superposition: one can show that the
minimizer of~\eqref{eqn:variational} can be written as
$\tens{u} = \tens{v} + \Ncal[\hooke:\tens{\epsilon}^p]$, where
$\tens{v}$ is a member of
\begin{equation}\label{eqn:solution_space}
  \Gamma({\tens{\epsilon}_p}) := \left\{\tens{v}\in \bar{H}^1(\body_p;\mathbb{R}^3) : g\big[\tens{v} - \Ncal [\tens{\epsilon}^p]\big] \geq 0 \text{ and } \div(\tens{\sigma}[\tens{v}]) = \tens{0}\right\},
\end{equation}
and minimises the potential energy written as surface integral
\begin{equation}\label{eqn:variational_surface}
  \inf_{\tens{v}\in\Gamma({\tens{\epsilon}_p})}\left\{\frac{1}{2}\int_{\partial\body}{\tens{T}[\tens{v}]\cdot\tens{v}\,\d S}\right\}.
\end{equation}
\begin{remark}\label{rem:modified_contact}
  Note the modified contact condition
  $g\big[\tens{v} - \Ncal[\hooke:\tens{\epsilon}^p]\big] \geq 0$,
  which explicitly accounts for the residual displacement at the
  surface caused by the plastic deformations. Since the contact
  problem is solved with fixed $\tens{\epsilon}^p$,
  \eqref{eqn:variational_surface} corresponds to an elastic contact
  problem with a modified contact surface
  $h_\text{mod} = h - \Ncal[\hooke:\tens{\epsilon}^p]\cdot
  \tens{e}_3\Big|_{\partial\body}$, a property that the
  elastic-plastic contact coupling algorithm exploits.
\end{remark}
Because of the previously mentioned property of displacements obtained
by our volume integral approach, $\Gamma$ can be readily approximated
using the span space of the interpolation functions underlying the
discrete Fourier transform. The minimization of the functional defined
in problem~\eqref{eqn:variational_surface} has been the subject of
extensive literature, and the reader is referred
to~\citep{stanley_fftbased_1997, polonsky_fast_1999,
  polonsky_numerical_1999, rey_normal_2017, campana_practical_2006,
  muser_meeting_2017} for various examples in the realm of normal
friction-less contact. Note that although we consider neither adhesion
nor friction in the elastic contact problem, they can without
difficulty be included in the global formulation (see e.g.\
\citet{wang_new_2013, pohrt_complete_2014,
  li_semianalytical_2003}). In this work, we use the modified
conjugate gradient devised by~\citet{polonsky_numerical_1999} to
solve~\eqref{eqn:variational_surface} in dual
form~\citep{kalker_variational_1977}, with the additional
condition~\eqref{eqn:total_pressure}.

\subsection{Plastic coupling}

The full elastic-plastic contact coupling scheme, developed
by~\citet{jacq_development_2002}, is given in
Algorithm~\ref{alg:plastic_coupling}. The algorithm leverages the
modified contact condition~\eqref{eqn:solution_space} (see
remark~\ref{rem:modified_contact}). The central loop's purpose is to
determine the increment of residual displacements
$\Delta\mathbb{U}_3^p$ at the surface of the elastic solid. These are
incorporated into the surface profile $\mathbb{H}$ for the elastic
contact solve step, during which plastic deformations are fixed. The
result of the elastic contact is a traction distribution $\mathbb{T}$
on $\partial\body$ that acts as a Neumann boundary condition to the
elastic-plastic problem. The solve step yields the total strain
increment $\Delta\mathbb{E}$, which is used to compute the residual
surface displacement. The convergence condition is established on the
evolution of $\Delta\mathbb{U}^p_3$ from one iteration to the
next. Finally, when convergence is reached within a specified
tolerance, the state $S$ is updated with the converged total strain
increment $\Delta\mathbb{E}$ and traction increment
$\Delta\mathbb{T}$.

\begin{algorithm}
  \caption{Elastic-plastic contact coupling
    algorithm~\citep{jacq_development_2002}.}\label{alg:plastic_coupling}
  \begin{algorithmic}
    \State Data: $W$ (normal load), $\mathbb{H}$ (surface profile),
    $S$ (current state), $\epsilon_\mathrm{tol}$ (tolerance),
    $N_\mathrm{max}$ (maximum iterations) \State
    $\Delta\mathbb{U}_{3,\text{prev}}^p \gets 0$ \State
    $\mathbb{T}_\mathrm{prev} \gets \tens{T}[S]$\MyComment{Previous
      tractions} \State $k \gets 1$ \State
    $\mathbb{H}_0 \gets \mathbb{H}$ \State
    $h_\text{norm} \gets \|\mathbb{H}_0\|$ \Repeat \State
    $\mathbb{T} \gets \mathbf{elastic\_contact}(W, \mathbb{H})$
    \MyComment{\citet{polonsky_numerical_1999} with FFT} \State
    $\Delta\mathbb{T} \gets \mathbb{T} - \mathbb{T}_\mathrm{prev}$
    \State
    $\Delta \mathbb{E} \gets \mathbf{plasticity}(\Delta\mathbb{T},
    S)$\MyComment{Strain increment s.t. $\Rcal[\Delta\mathbb{E}] = 0$}
    \State
    $\Delta \mathbb{U}_3^p \gets
    \Ncal[\hooke:\Delta\tens{\epsilon}_p(\Delta\mathbb{E},
    S)]\big|_{\partial\body}\cdot \tens{e}_3$\MyComment{Surface
      residual displacement increment} \State
    $e \gets
    \|\Delta\mathbb{U}_3^p-\Delta\mathbb{U}_{3,\text{prev}}^p\|\,/\,
    h_\text{norm}$\MyComment{Error} \State
    $\Delta \mathbb{U}_{3,\text{prev}}^p \gets \Delta
    \mathbb{U}_{3,\text{prev}}^p + \lambda (\Delta \mathbb{U}_{3}^p -
    \Delta \mathbb{U}_{3, \text{prev}}^p)$ \MyComment{Relaxation}
    \State
    $\mathbb{H} \gets \mathbb{H}_0 - \Delta
    \mathbb{U}_{3,\text{prev}}^p$ \State
    $\Delta\mathbb{U}_{3,\text{prev}}^p \gets \Delta\mathbb{U}_{3}^p$
    \State $k \gets k + 1$ \Until
    $e < \epsilon_\mathrm{tol} \vee k > N_\mathrm{max}$ \State
    $\mathbf{update}(S, \Delta\mathbb{E})$\MyComment{Increment $e_p$
      and $\tens{\epsilon}_p$}
  \end{algorithmic}
\end{algorithm}

The user of Algorithm~\ref{alg:plastic_coupling} is free to use any
elastic contact solver and non-linear plasticity solver as drop-ins
for $\mathbf{elastic\_contact}$ and $\mathbf{plasticity}$. For the
simulations presented in this work, the elastic contact solver uses
the projected conjugate gradient proposed
by~\citet{polonsky_numerical_1999}, with an FFT approach for the
gradient computation\footnote{Note that the influence coefficients of
  those methods are based on the \citet{westergaard_bearing_1939} and
  \citet{johnson_contact_1985} solutions but match the expression we
  give for the Boussinesq tensor in the partial Fourier
  space.}~\citep{stanley_fftbased_1997, rey_normal_2017}. For the
non-linear plastic solver, we use the DF-SANE algorithm described
above. The relaxation parameter $\lambda$ can take values in
$\rbrack 0, 1\rbrack$ and helps the algorithm to converge on large
loading steps.


\section{Method validation}\label{sec:validation}

We now present two validation steps for our computational
methodology. First, evaluations of integral operators $\Ncal$ and
$\tens{\nabla}\Ncal$ by our approach are compared in
Section~\ref{sec:sphere} to the analytical solution of a hydrostatic
eigenstrain in a spherical inclusion embedded in a half-space. Then,
in Section~\ref{sec:ep_validation}, the entire elastoplastic contact
method is assessed using comparisons to literature results and a
reference finite-element simulation.

\subsection{Validation of the Mindlin operators}\label{sec:sphere}

The displacement generated by a constant hydrostatic eigenstrain
$\tens{\epsilon} = \alpha T \Id$ applied in a spherical region
embedded in a half-space is given by $\tens{u}=\Ncal[\tens{w}]$ with
$\tens{w}=\hooke:\tens{\epsilon} = 2\mu\alpha
T\tfrac{1+\nu}{1-2\nu}\tens{I}$, and has been derived analytically in
closed form by \citet{mindlin_thermoelastic_1950}. We use that
reference solution to validate our implementation of operators $\Ncal$
and $\tens{\nabla}\Ncal$. Let the inclusion support be the ball of
radius $a$ and center $(0,0,c)$ (with $c>a$). The surface
displacements $u_r$ and $u_{z}$ are given, using cylindrical
coordinates $(r, \theta, z)$, by:
\begin{align}
  u_r(r) = \frac{4a^3}{3R_1^3}\beta(1-\nu)r,\qquad
  u_z(r) = -\frac{4a^3}{3R_1^3}\beta(1-\nu)c,
\end{align}
with $R_1 = \sqrt{r^2 + {c}^2}$ and
$\beta = \alpha T\frac{1+\nu}{1-\nu}$. For this example, we use $c=2a$
and $\nu=0.3$, while the periodicity cell $\Bcal_p$ and its
discretization $\tens{N}$ are defined by
$\Bcal_p=\rbrack-15a,15a\lbrack^2\times\rbrack0,10a\lbrack$ and $\tens{N}=(128,128,126)$.
Figures~\ref{fig:mindlin_cheng_disp}a
and~\ref{fig:mindlin_cheng_disp}b show the horizontal displacement
$u_r$ and the vertical displacement $u_z$, respectively, evaluated
along the $x_1=x_3=0$ line. The computed and reference displacements
are in good agreement in the central zone, the expected distortion
induced by the periodic boundary conditions becoming apparent only
close to the boundary of $\Bcal_p$.

\begin{figure}
  \centering \includegraphics[draft=false]{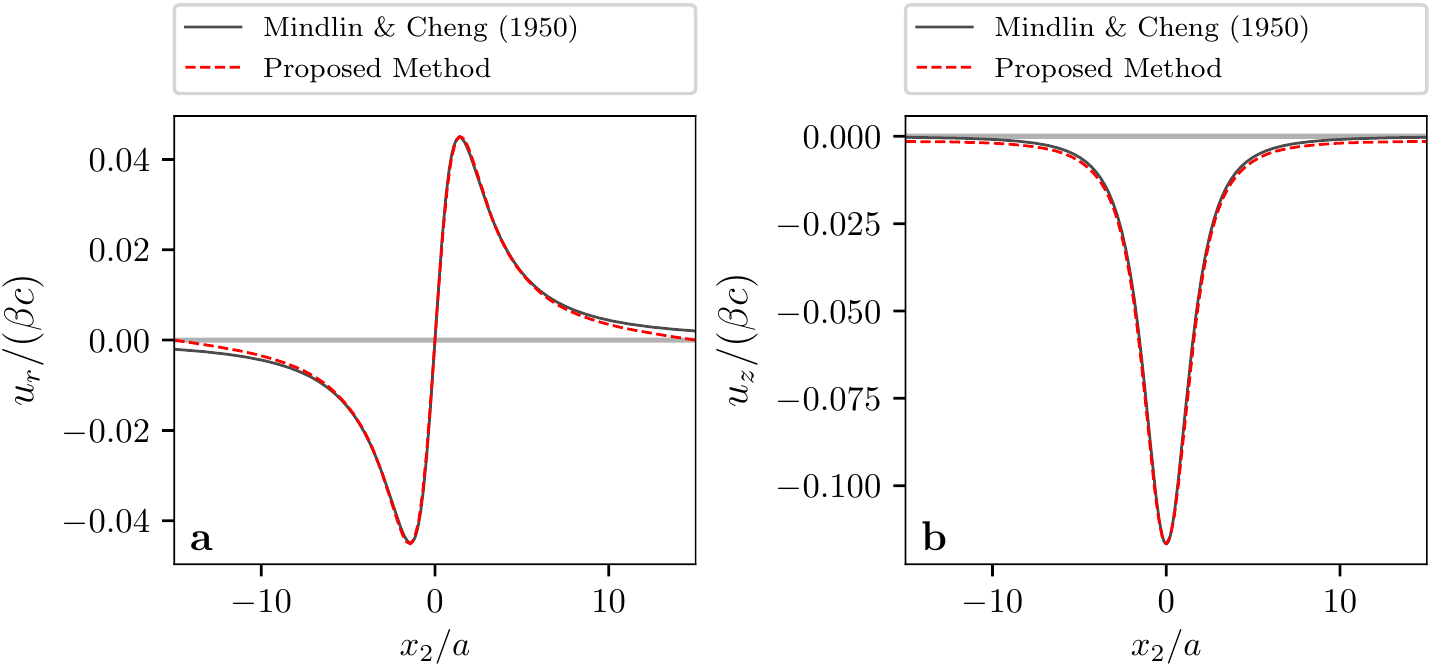}
  \caption{\textbf{Normalized surface displacements due to a
      hydrostatic spherical inclusion in a half-space}. Comparison
    between the~\citet{mindlin_thermoelastic_1950} solution and the
    proposed VIM.\@The displacements are shown along the
    $x_1 = x_3 = 0$ line. The agreement between the analytical and
    computed solutions is very good in the central part of the
    periodicity cell, while distortion induced by periodic conditions
    becomes apparent near its boundary.}\label{fig:mindlin_cheng_disp}
\end{figure}

Recalling now that
$\tens{\sigma}=\hooke:\tens{\nabla}\Ncal[\tens{w}]-\tens{w}$,
Fig.~\ref{fig:mindlin_cheng_stress} shows the validation of
$\tens{\nabla}\Ncal$ by way of a comparison of stresses
$\sigma_\theta$, and $\sigma_z$ produced by the inclusion, whose
values along the vertical line $r=0$ going through the inclusion
center are shown in Figs.~\ref{fig:mindlin_cheng_stress}a
and~\ref{fig:mindlin_cheng_stress}b, respectively. The relevant
analytical values are (for $r=0$):
\begin{align}
  \sigma_\theta
  &= \frac{2\mu\beta a^3}{3}\Big( \frac{4\nu-3}{{(z\shp c)}^3}+\frac{6c}{{(z\shp c)}^4} + \frac{1}{{|z\shm c|}^3} \Big)
  &
    \sigma_z
  &= \frac{2\mu\beta a^3}{3}\Big( \frac{6z+c}{{(z\shp c)}^4} - \frac{2}{|z\shm c|^3} \Big)
  && z\in[0,a\lbrack\cup]3a,+\infty\lbrack,\\
  \sigma_\theta
  &= \frac{2\mu\beta a^3}{3}\Big( \frac{4\nu-3}{{(z\shp c)}^3}+\frac{6c}{{(z\shp c)}^4} - \frac{2}{a^3} \Big),
  &
    \sigma_z
  &= \frac{2\mu\beta a^3}{3}\Big( \frac{6z+c}{{(z\shp c)}^4}- \frac{2}{a^3} \Big),
  && z\in\rbrack a,3a\lbrack
\end{align}
with $\beta = \alpha T\frac{1+\nu}{1-\nu}$. We observe on
Figures~\ref{fig:mindlin_cheng_stress}a
and~\ref{fig:mindlin_cheng_stress}b a Gibbs effect at the inclusion
boundary, caused by the Fourier approximation of the discontinuity of
the eigenstrain, as mentioned in
Section~\ref{sec:spectral_interpolation}. This should however not have
a significant effect on elastoplastic simulations, as plastic
deformations should be continuous provided there is no shear
band. Nonetheless, our method accurately represents the large
discontinuity in the tangential stress, see
Figure~\ref{fig:mindlin_cheng_stress}b.

\begin{figure}[t]
  \centering \includegraphics[draft=false]{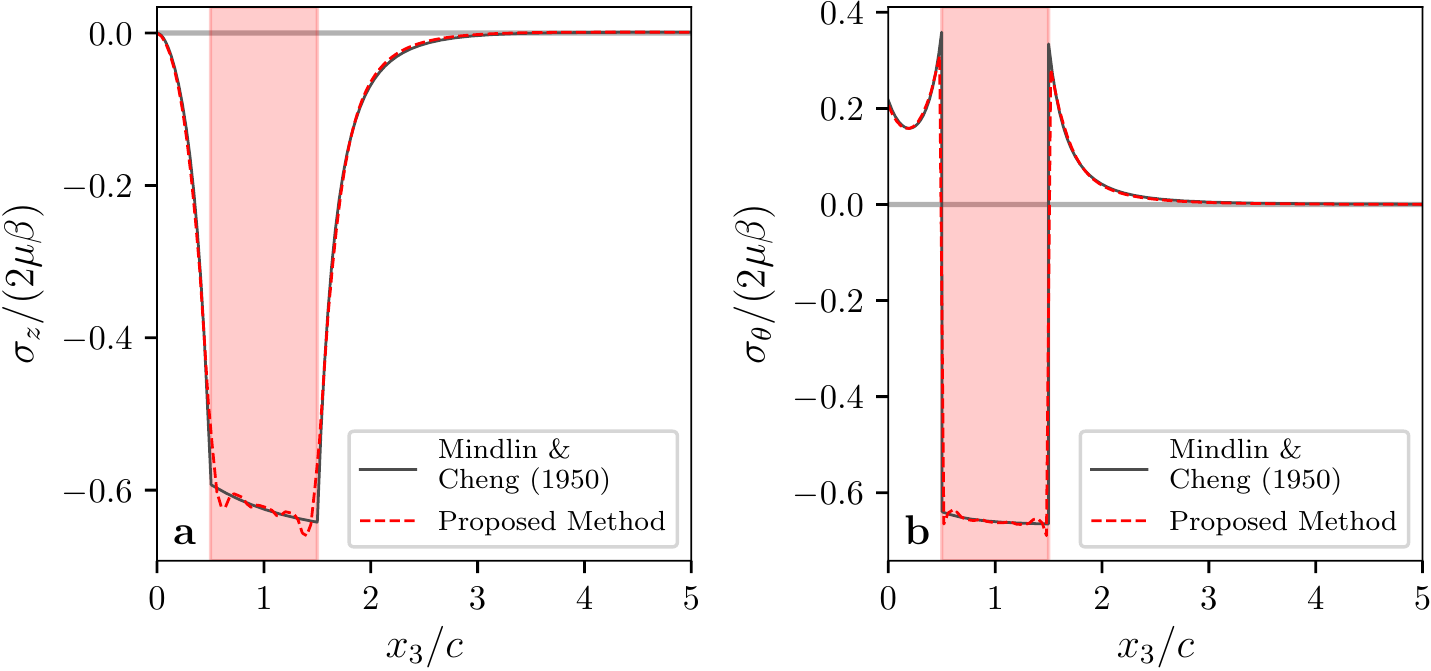}
  \caption{\textbf{Normalized stresses due to a hydrostatic spherical
      inclusion in a half-space}. Comparison between
    the~\citet{mindlin_thermoelastic_1950} solution and the proposed
    volume integral method. The stresses are shown along the $x_3$
    axis. The red region is where the eigenstrain
    $\tens{\epsilon} = \alpha T \Id$ is imposed. One can observe a
    good agreement of the numerical approximation with the analytical
    solution outside the inclusion. Some Gibbs effect can be observed
    at the boundary of the inclusion, with spurious oscillations in
    the inclusion. This is due to the Fourier approximation of the
    discontinuous eigenstrain
    function.}\label{fig:mindlin_cheng_stress}
\end{figure}

\subsection{Elastoplastic contact validation}\label{sec:ep_validation}

We validate our complete proposed method (including elastoplastic
contact) against an axi-symmetric FEM analysis
by~\citet{hardy_elastoplastic_1971} which provides the surface
pressure distribution for a rigid spherical indenter on an
elastic-perfectly plastic material. For Hertzian contact with
$\nu = 0.3$, the maximum shear stress occurs at a depth
$x_{3} \approx 0.57a$, with $a$ the contact radius. Moreover, the von
Mises stress reaches $\sigma_{\text{y}}$ for a maximum surface
pressure of
$p_{\text{y}} =
1.6\sigma_{\text{y}}$~\citep{johnson_contact_1985a}. From this, one
can compute the total load $W_{\text{y}}$ and contact radius
$a_{\text{y}}$ at the onset of yield~\citep{hardy_elastoplastic_1971}:
\begin{align}
  \text{(a) \ }W_{\text{y}} = \frac{\pi^3 R^2}{6 {E^{\star}}^2}p_{\text{y}}^3,\quad\quad
  \text{(b) \ }a_{\text{y}} = \sqrt{\frac{3W_{\text{y}}}{2\pi p_{\text{y}}}},
\end{align}
where $R$ is the indenter radius and $E^{\star} := E / (1-\nu^2)$ is
the Hertz elastic modulus. We also define
$\tau_{\text{y}} := \sigma_{\text{y}} / \sqrt{3}$ as the shear yield
stress. $W_{\text{y}}$, $a_{\text{y}}$ and $\tau_{\text{y}}$ are used
to normalize loads, lengths and stresses
respectively. Figure~\ref{fig:hardy} compares the surface pressure
$p[\tens{u}]$ computed using Algorithm~\ref{alg:plastic_coupling} with
the corresponding values obtained by~\citet{hardy_elastoplastic_1971}
(dashed lines), for different load ratios $W / W_{\text{y}}$. The dark
continuous line is the circumferential average of the pressure, while
the lighter zone shows the maximum and minimum pressure for a given
radial coordinate $r$. The difference between the maximum and minimum
pressures at the edge of contact is due to the discretization and the
periodic boundary conditions, which make the contact region shape
deviate from a disk. The simulation parameters are given
in~\ref{app:hardy_comparison}. As mentioned
in~\citep{hardy_elastoplastic_1971}, normalized results are
independent of the $E / \tau_{\text{y}}$ ratio.\enlargethispage*{5ex}

\begin{figure}[t]
  \centering \includegraphics[draft=false]{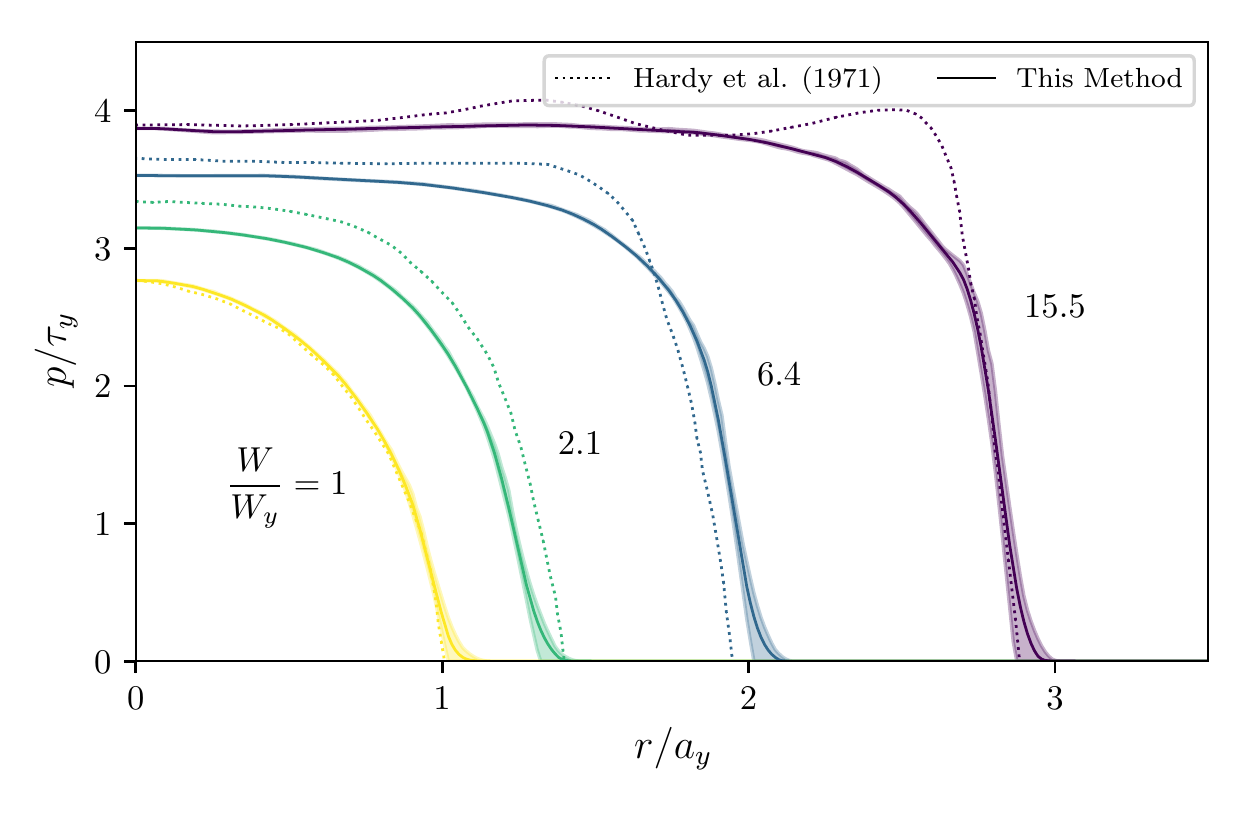}\vspace*{-4ex}
  \caption{\textbf{Elastic perfectly-plastic Hertzian contact,
      comparison with~\citet{hardy_elastoplastic_1971}}. Increase of
    the applied load beyond the initial yield shows that the pressure
    profile deviates from the elliptic Hertzian profile by flattening
    of the curve at the axis of symmetry, with a plateau whose extent
    increases with the load. The results
    of~\citet{hardy_elastoplastic_1971} however show oscillations of
    the pressure profile at high plastic loads which is not reproduced
    by our simulation. As there is, to our knowledge, no physical
    reason to these oscillations, they are likely due to the coarse
    discretization of the finite-element mesh they used in their
    study\protect\footnotemark. The simulation parameters are given
    in~\ref{app:hardy_comparison}}\label{fig:hardy}
\end{figure}
\footnotetext{Their stiffness matrix fits in the 512K RAM of the
  IBM/360 used for simulations, which is impressive for 1971. The
  smoothness of the results (extracted from figure 5) is likely due to
  the figure being drawn by hand.}

\begin{figure}[t]
  \centering \includegraphics[draft=false]{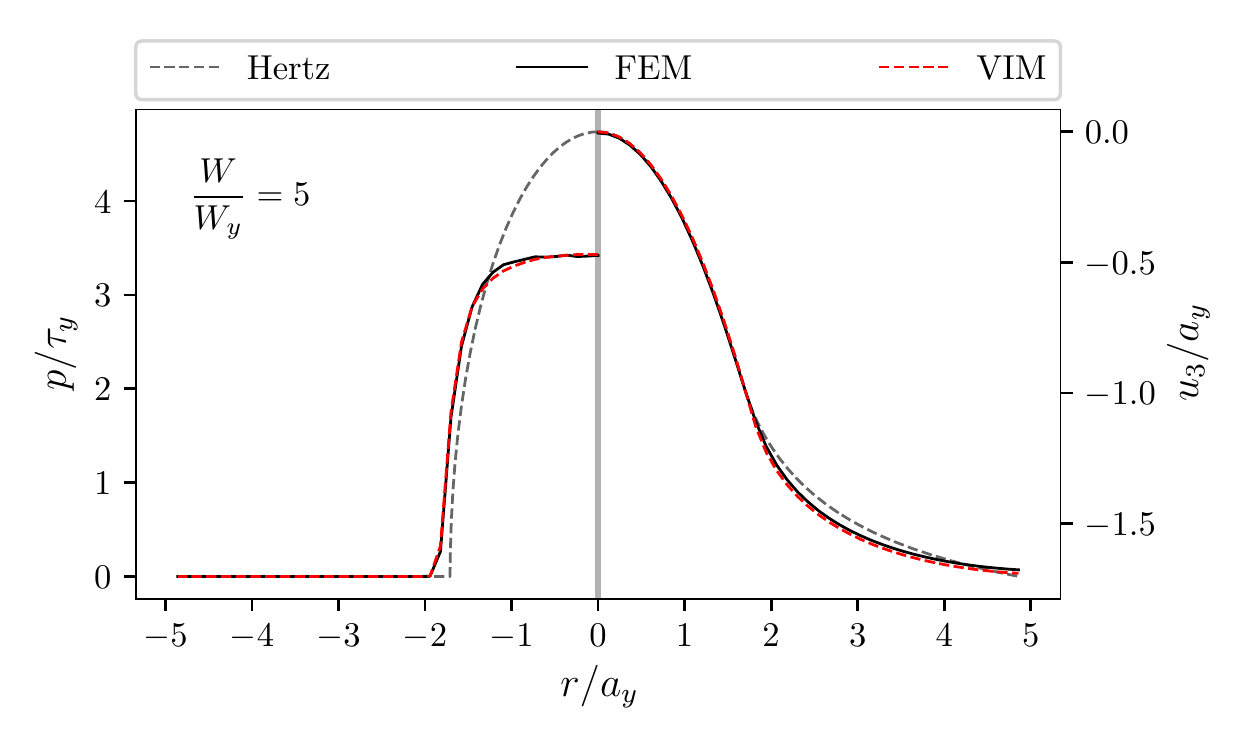}\vspace*{-4ex}
  \caption{\textbf{Elastoplastic Hertzian contact}. Comparison between
    the elastic Hertz solution~\citep{johnson_contact_1985a}, a
    simulation using algorithm~\ref{alg:plastic_coupling} where the
    surface residual displacement is calculated by FEM, and a full VIM
    simulation. Both simulations have identical surface
    discretization. As expected~\citep{johnson_experimental_1968}, the
    plastic pressure distribution deviates from the elliptical shape
    of the Hertzian distribution. The contact radius is larger in the
    plastic case. The simulation parameters are given
    in~\ref{app:akantu_comparison}.}\label{fig:plastic_hertz}
\end{figure}

Both solutions show a flattening of the initially-ellipsoidal pressure
distribution near the axis of symmetry (also observed
in~\citet{johnson_experimental_1968, jacq_development_2002,
  wang_new_2013}), with a plateau that extends as the load is
increased. There are however significant differences between the two
sets of results: in the plastic range, the agreement on both the
contact radius and the plateau value is poor. Also, the results of
\citet{hardy_elastoplastic_1971} feature oscillations in the pressure
profile at the highest loading increment, which are likely due to the
coarse mesh that was used, combined with large loading
steps~\footnote{Our own experiments with finite-element simulations
  exhibited the same behavior, as well as disagreement in the plateau
  value, for too-large element sizes.}.

To assess these effects and confirm the latter remark, we have
effected additional comparisons, this time to an implementation of
algorithm~\ref{alg:plastic_coupling} where the strain increment is
computed using a first-order FEM approach with the open-source
high-performance code Akantu~\citep{richart_implementation_2015}, the
rest of the algorithm using the same code as for the VIM. The
geometry, material properties, loading and discretization are as given
in~\ref{app:akantu_comparison}. Figure~\ref{fig:plastic_hertz} shows
the surface pressure and vertical displacement for a load ratio of
$W/W_{\text{y}} = 5$, the corresponding Hertz elastic solution being
also plotted for additional comparison. We observe an excellent
agreement between the FEM and VIM solutions. In addition,
Figure~\ref{fig:plastic_deformations} shows good qualitative agreement
between the values of $|\tens{\epsilon}^p|$ on a symmetry plane
obtained for the same load with the FEM and the VIM.

\begin{figure}
  \centering
  \begin{tikzpicture}
    \node(0, 0) {\includegraphics[draft=false,
      width=0.8\textwidth]{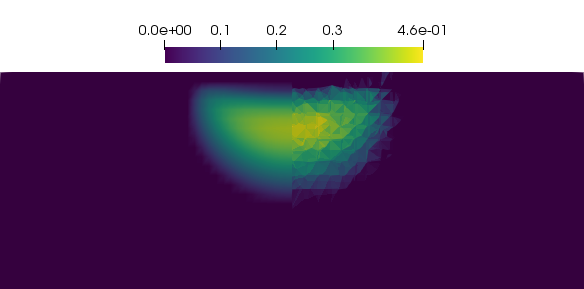}};
    \draw[white, very thick, dashed]
    (0, 1.7) -- ++(0, -5);
    \draw(-1, -2.5) node[white]{\textbf{VIM}};
    \draw(1, -2.5) node[white]{\textbf{FEM}};
    \draw(0, 3) node[centered, scale=1.2]{\textbf{Plastic strain norm}};
  \end{tikzpicture}
  \caption{\textbf{Plastic strain norm $|\tens{\epsilon}^p|$ in
      elastoplastic Hertzian contact}. Results obtained with
    algorithm~\ref{alg:plastic_coupling}, with the plastic problem
    solved using our VIM approach (left) or a first order FEM
    (right). Note that in the FEM case, plastic deformations are
    piece-wise constant, whereas in the volume integral result they
    are interpolated between nodal values (cf.\
    Section~\ref{sec:discretized_operators}). Nonetheless, the
    solutions give similar plastic zone size and maximum plastic
    deformation norm. The simulation parameters are given
    in~\ref{app:akantu_comparison}. }\label{fig:plastic_deformations}
\end{figure}
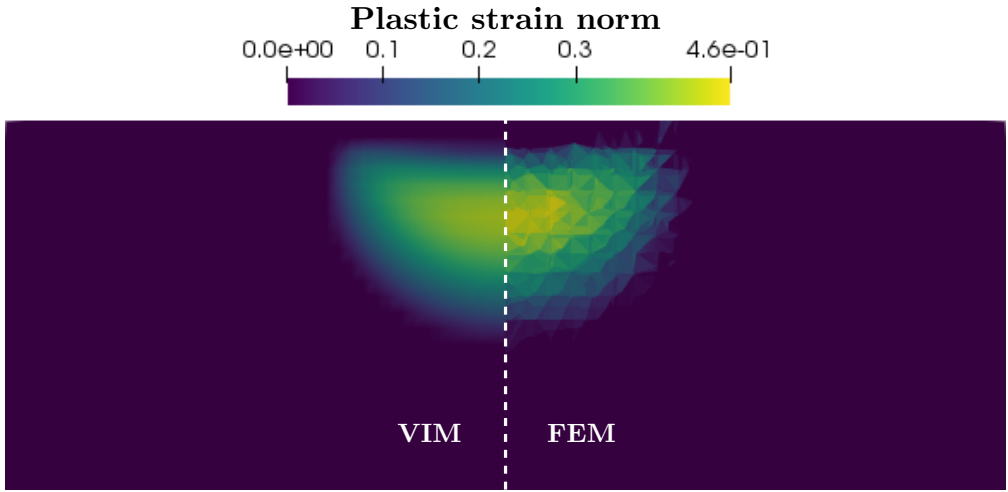


\section{Algorithmic complexity}\label{sec:complexity}
We now compare the computational cost of the application of the
operator $\Ncal$ to that of an elastic solve step of a finite elements
simulation involving the same number $N=N_1 N_2 N_3$ of
nodes. Irrespective of the numerical method used, the optimal
complexity would be $\Ocal(N)$. In the proposed methodology, the
evaluation of $\Ncal[\tens{w}]$ for given $\tens{w}$ is decomposed
into two computational steps: the multiple 2D fast-Fourier transforms
and the computation of equation~\eqref{eqn:typical_integral}. Assuming
$N_1 = N_2 = N_3$, the former has a complexity of
$\Ocal(N_1N_2N_3\log(N_1N_2)) \sim \Ocal(N\log(N))$, while the latter
has a complexity of $\Ocal(N_1N_2N_3^2) \sim \Ocal(N^{4/3})$ with a
naive implementation. Using the cutoff proposed
in~\ref{sec:x_3_discrete} brings the algorithmic complexity of the
evaluations of~\eqref{eqn:typical_integral} down to $\Ocal(N)$
(see~\ref{sec:app-complexity}), so the asymptotic cost of evaluating
$\Ncal[\tens{w}]$ is $\Ocal(N\log(N))$. For the direct solve step of a
finite element elastic simulation with $N$ nodes, the algorithmic
complexitiy of a sparse Cholesky factorization is
$\Ocal(N^{3/2})$\footnote{This bound is established for 2D finite
  elements, but we expect it to remain representative in the present
  3D case.}~\citep{lipton_generalized_1979}. We compare in
figure~\ref{fig:complexity} the relative computation times of the
application of $\Ncal$ and the elastic solve step of Akantu, which
uses the direct solver package MUMPS~\citep{amestoy_fully_2001} to
perform the factorization. A regular mesh with $N$ nodes was
used\footnote{Note that actually applying the FEM to the present
  mechanical problem would in fact require discretization of a much
  larger domain to model the fields away from the potentially plastic
  zone.}. We can observe that for large problem sizes the computation
times fit the theoretical asymptotic complexities, showing the clear
advantage of the proposed VIM over FEM. One should also note that
memory needed for the factorization of the stiffness matrix for
$2^{21}$ nodes was larger than 128~GB whereas the VIM only required
1.27~GB for this case. For large problems, both the memory imprint
and (absolute) computing time are two orders of magnitude smaller with
the proposed approach.

As a closing remark, evaluating $\Ncal$ in physical space and without
any acceleration method would entail a $\Ocal(N^2)$ complexity, making
its use unrealistic for 3D problems. This complexity can be brought
down to $\Ocal(N)$ by means of a multi-level fast multipole (ML-FM)
approach. However, implementing the latter is quite involved in
general, and here would require intricate and expensive close-range
numerical quadrature methods for dealing with the strong singularity
and complex expressions of the kernel in the physical space. This may
explain why ML-FM treatments of VIMs have received only limited
attention to date.

\begin{figure}[t]
  \centering \includegraphics{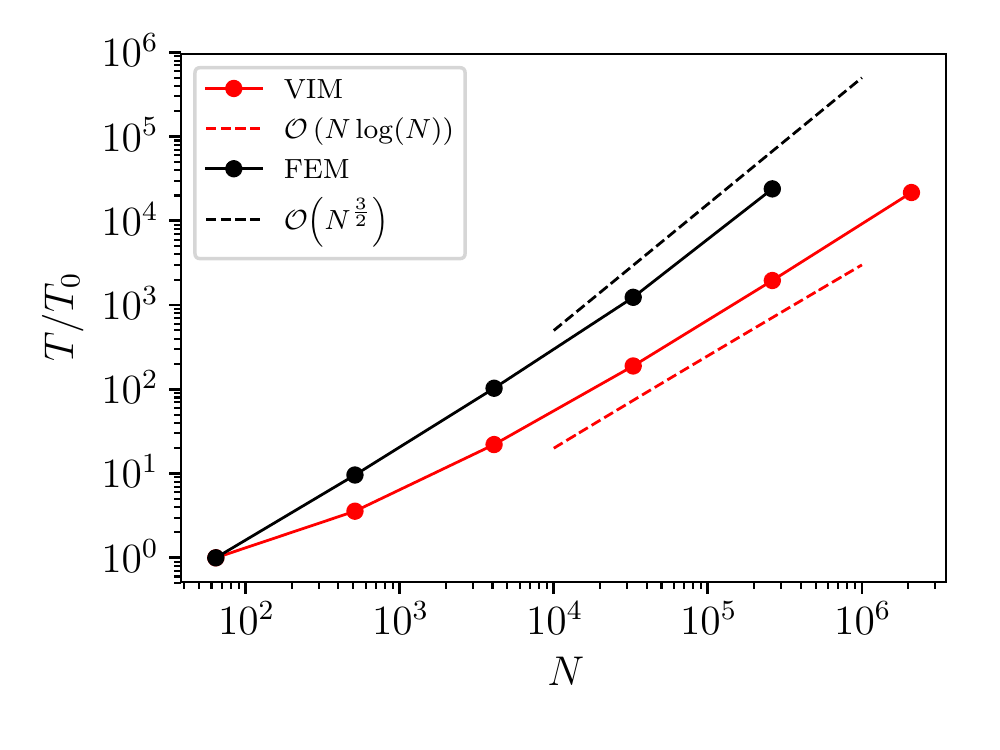}\vspace*{-4ex}
  \caption{\textbf{Relative computation times for VIM and FEM.} We
    compare the application of the operator $\Ncal$ to an elastostatic
    FEM solve step (Cholesky factorization). The reference time $T_0$
    is $T_0 = 1.41\cdot 10^{-3}$s for the VIM and
    $T_0 = 1.77\cdot 10^{-2}$s for the FEM. The scaling for large
    problem sizes agrees with the theoretical algorithmic
    complexities. For $N = 2^{21}$ the stiffness matrix factorization
    needed over 128~GB of memory, an amount two orders of magnitude
    larger than what is required for the VIM.}\label{fig:complexity}
\end{figure}



\section{Rough surface contact}\label{sec:rough}

We now showcase the applicability of the method to rough surface
contact. Self-affine surfaces are commonly used in this
context~\citep{yastrebov_contact_2012, muser_meeting_2017}. They are
characterized by a power-spectral density obeying a power
law~\citep{persson_nature_2005}. For this simulation, we generate a
random rough surface using an FFT filter
algorithm~\citep{hu_simulation_1992}, specifying a Hurst exponent of
0.8, with the largest wavelength $\lambda_l$, the roll-off wavelength
$\lambda_r$ and the smallest wavelength $\lambda_s$ given by
$L / \lambda_l = 1$, $\lambda_l / \lambda_r = 2$ and
$\lambda_s / \Delta l = 4$ (cf.\ \ref{app:rough_surface} for graphical
representation of the PSD) where $L := L_1 = L_2$ and
$\Delta l = L / N$ (equal number of points $N = 512$ in $x_1$ and
$x_2$). The $x_3$ domain $L_3 = 0.4\cdot L$ is discretized in
$N_3 = 64$ points, yielding upwards of 100 million degrees of freedom
for the plasticity problem (where $\Delta\tens{\epsilon}$ is the
primary unknown) solved on a single compute node. The Poisson
coefficient is set\footnote{These values do not correspond to any
  specific material and were arbitrarily selected for the purpose of
  showcasing the effect of plasticity.} to $\nu = 0.3$, the yield
ratio to $\sigma\ysub / E = 10^{-1}$ and the hardening modulus to
$E\hsub = 10^{-2}E$.

\begin{figure}[t]
  \centering \includegraphics[draft=false]{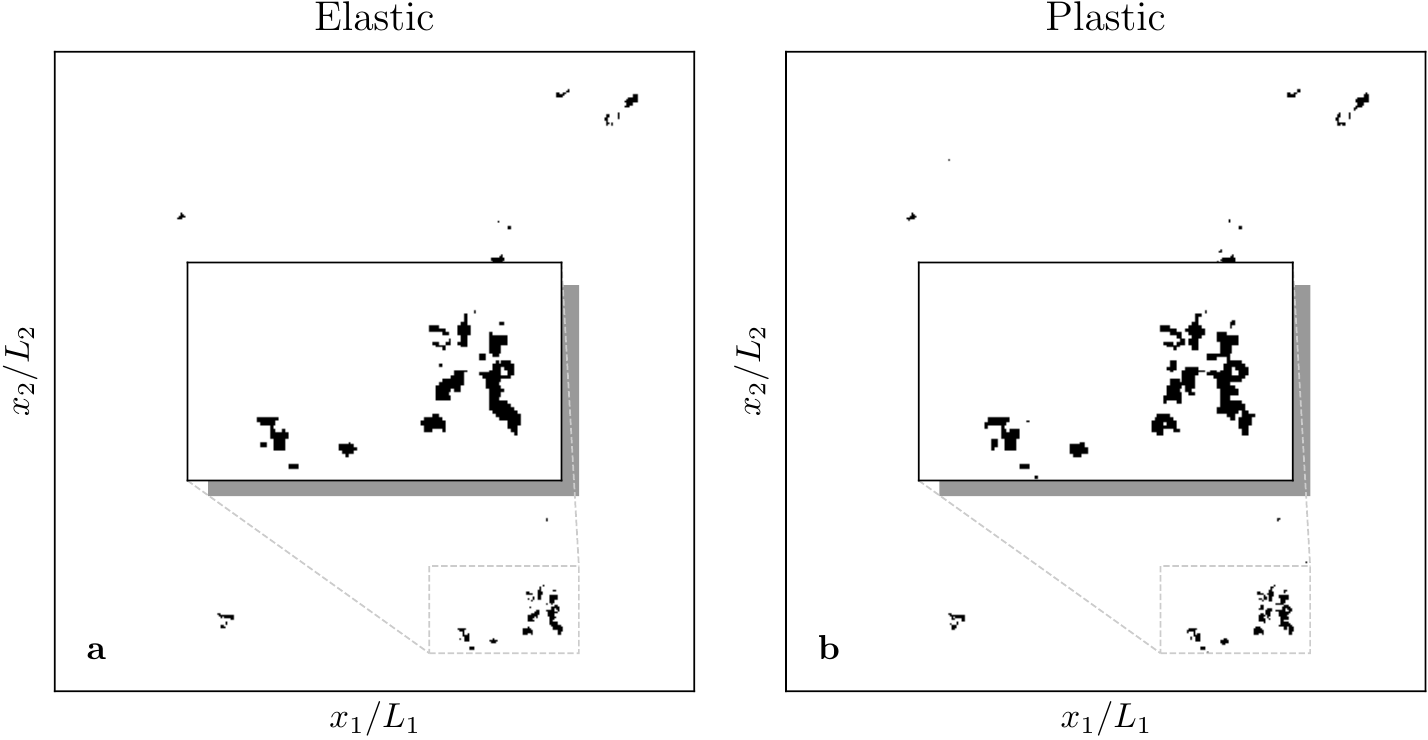}
  \caption{\textbf{Contact area comparison in rough surface
      contact}. Both simulations have the same contacting rough
    surface. Figure a shows the elastic contact area, with the inset
    zooming on the lower left cluster of micro-contacts. Figure b
    shows the contact area obtained for the same load in the case of
    elastic-plastic contact, which is 24\% larger than in the elastic
    case. The details of the micro-contacts are also different:
    disjoint elastic contacts have coalesced in the plastic case, and
    some new contacts have appeared. This may change the micro-contact
    distribution, influencing phenomena such as
    wear.}\label{fig:rough}
\end{figure}

Figure~\ref{fig:rough} shows a comparison between the actual contact
areas produced by elastic and elastoplastic contact simulations with
the same applied load
$W = 10^{-3}E \sqrt{\langle |\nabla h|^2 \rangle}$ (normalization by
the root-mean-square of slopes collapses the load to contact area
relation in the elastic case,
see~\citealp{hyun_finiteelement_2004}). It should be mentioned that
algorithm~\ref{alg:plastic_coupling} is not particularly well suited
for rough surface contact: it is by default not robust to large load
steps and requires tuning of the relaxation parameter, which renders
the convergence slow and very dependent on discretization and on
$\sigma\ysub / E$. The lack of robustness of the algorithm prevents
full use of the capacities of Newton-Krylov and spectral residual
solvers, which usually allow for large load steps. To our knowledge,
this has not yet been addressed, as recent work on elastic-plastic
contact with VIMs also use alternating projection approach of
algorithm~\ref{alg:plastic_coupling}~\citep{amuzuga_fully_2016,
  zhang_contact_2018}. Finding a coupling algorithm better suited than
algorithm~\ref{alg:plastic_coupling} to rough surface contact will be
part of our future work on the subject. For instance, the
elastoplastic contact problem considered here could be recast as a
second-order conic program. This suggests the possibility of resorting
to interior point methods, which have proved robust and efficient in
other contexts also involving large-scale simulations for non-linear
materials~\citep{bleyer_advances_2018}.

Hence, we have confined this preliminary study to small contact areas
and a rather large $\sigma\ysub / E$ ratio (which would correspond to
the behavior of elastomers). Noticeable differences are nonetheless
observed between the elastic and elastic-plastic solutions: with
plasticity, the contact area is 24\% larger (for the same load) and
details of micro-contacts are qualitatively different. The insets of
figure~\ref{fig:rough} show that some micro-contacts coalesce due to
plasticity while others grow larger. This may have implications on the
distribution of micro-contact areas, which plays a crucial role in
wear~\citep{frerot_mechanistic_2018}.

\begin{figure}
  \centering
  \includegraphics[width=0.6\textwidth,draft=false]{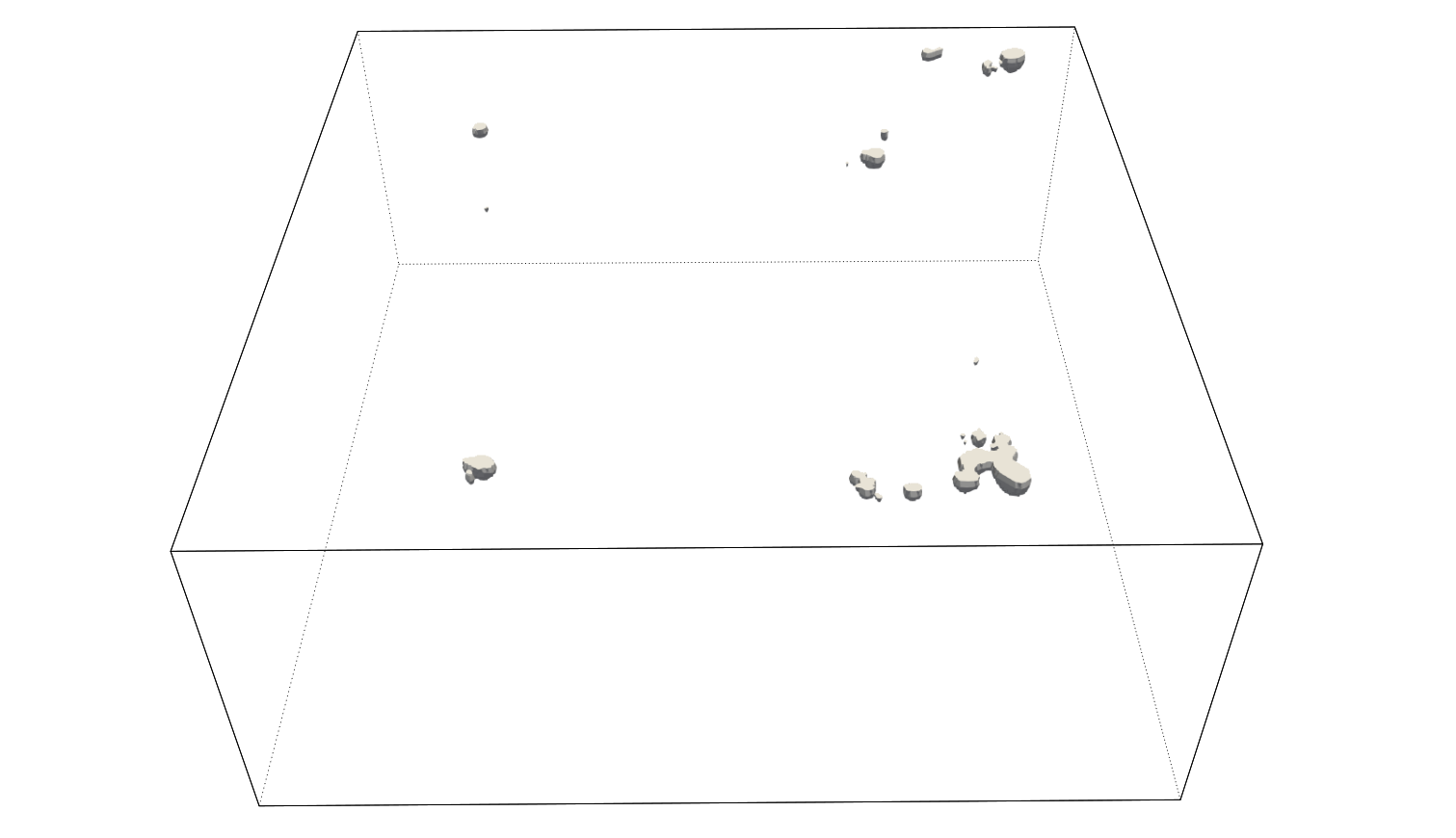}
  \caption{\textbf{Plastic volumes in material bulk}. As expected,
    they are located beneath contact zones (see
    figure~\ref{fig:rough}). Even for low contact area situations, the
    plastic activity reaches the contact
    surface.}\label{fig:plastic_volumes}
\end{figure}

Figure~\ref{fig:plastic_volumes} shows the plastic volumes in the bulk
of the material. Even for such a small contact area, plastic activity
in the subsurface material is enough to make disjoint micro-contacts
share the same plastic deformation zones. As the proposed method
allows an unprecedentedly fine representation of the plastic
deformations, one can exploit this data to compute statistics on the
distribution of plastic volumes.


\section{Conclusion}
We have developed a new, high-efficiency spectral method for solving
elastoplastic contact problems. It relies on a novel derivation of the
Mindlin fundamental solution directly in the 2D partial Fourier
coordinate space, allowing memory savings, error reduction and
acceleration of the remaining one-dimensional convolution in a
periodic setting, while exploiting the performance of the fast-Fourier
transform. We also show the mathematical soundness of the different
steps in the method (with the theory of distributions), justifying the
discretization of the integral operators which was simply assumed in
previous works on spectral methods in contact.

The performance of the method comes from the periodic nature of the
problem and its intimate relationship with the discrete Fourier
transform. In the investigation of physical phenomena, having a
periodic system is not a problem as long as surfaces are
representative, i.e.\ their largest wavelength is several times
smaller than the size of the
system~\citep{yastrebov_contact_2012}. This entails a large number of
points in the discretized system, which our method is capable of
handling while still considering the non-linear and volumetric nature
of plastic deformations. The latter can be resolved with an accuracy
fine enough for computation and statistics of plastic volumes, which
have until now never been studied in detail. As comparison, a
simulation of 100 million degrees of freedom with a finite-element
approach would require distributing the computational load over
several cluster nodes, whereas we have shown that our approach can
handle such a simulation on a single node (both the required memory
and computing time being two orders of magnitude smaller). The
performance of the global elastoplastic contact solver is currently
the bottleneck of the presented method and renders a systematic study
of rough surface contact difficult. Our next efforts will be to
replace the solver by an interior-point method, which have proven
efficient in a similar class of problems~\citep{bleyer_advances_2018}.

As mentioned by~\citet{vakis_modeling_2018}
and~\citet{muser_meeting_2017}, there are many tribological phenomena
that will benefit from accurate results of elastic-plastic rough
contact for which this method can be instrumental.

\section*{Acknowledgments}
L. F., G. A. and J.-F. M. acknowledge the financial support of the
Swiss National Science Foundation (grant \# 162569 ``Contact mechanics
of rough surfaces''). L. F. acknowledges the help of Claire Capelo for
the elastoplastic contact FEM code.


\FloatBarrier%
\setcounter{figure}{0}

\begin{appendix}

  \section{Invertibility and linear independence of
    \texorpdfstring{$\tens{A}^+$}{\textbf{A}+} and
    \texorpdfstring{$\tens{A}^-$}{\textbf{A}-}}\label{app:invertibility}

  Due to its structure, $\tens{A}^\pm$ has two unit eigenvalues. The
  determinant of $\tens{A}^\pm$ is therefore:
  \begin{align}
    \det(\tens{A}^\pm)
    = \text{trace}(\tens{A}^\pm)
    &= e^{\mp qy_3}\left(\frac{c}{c+2}qy_3 (\tens{\Delta}^\pm\cdot\tens{\Delta}^\pm) + 1\right)\\
    &= e^{\mp qy_3}\left(\frac{c}{c+2}qy_3\left(1-{\left(\frac{q_1}{q}\right)}^2 - {\left(\frac{q_2}{q}\right)}^2\right) + 1\right)
      = e^{\mp qy_3},
  \end{align}
  so does not vanish
  $\forall (\tens{q}, y_3) \in
  \mathbb{R}^2\times\mathbb{R}$. Therefore $\tens{A}^\pm$ is
  invertible and $\mathrm{rank}(\tens{A}^\pm) = 3$. The linear
  independence of $\tens{A}^+$ and $\tens{A}^-$ then stems from their
  being proportional to different exponential functions.

  \section{Proof of
    Theorem~\ref{thm:periodic_convolution}}\label{nonper:to:per}

  The displacement $\tens{u}=\Ncal[\tens{w}]$ ($\Ncal$ being the
  Mindlin integral operator) solves the problem
  \begin{equation}
    \text{(a) \ }\navier[\tens{u}] = -\text{div}\tens{w} \quad \text{in }\body, \qquad
    \text{(b) \ }\tens{T}[\tens{u}]=\tens{0} \quad \text{on }\partial\body \label{eigenstress:PDE}
  \end{equation}
  In the non-periodic case, the partial Fourier version of the above
  problem is the ODE system
  \begin{equation}
    \text{(a) \ }\fnavier(\tens{q})\cdot\ftens{u}(\tens{q},y_3)
    = -\ftens{\nabla}\cdot\ftens{w}(\tens{q},y_3) \quad \big(y_3\in\lbrack 0,\infty\lbrack \big), \qquad \qquad
    \text{(b) \ }\ftens{T}(\tens{q})\cdot\ftens{u}(\tens{q},x_3) =  \tens{0} ,
    \label{eigenstress:ODE}
  \end{equation}
  where $\tens{q}\in\Rbb^2$ acts as a
  parameter. Solving~\eqref{eigenstress:ODE} for arbitrary sources
  $\ftens{w}(\tens{q},y_3)$ yields
  $\ftens{u}(\tens{q},\smallbullet)=\widehat{\Ncal[\tens{w}]}(\tens{q},\smallbullet)$
  for any $\tens{q}\in\Rbb^2$, $\widehat{\Ncal[\tens{w}]}$ being given
  by~\eqref{N:conv:Fourier} with $\ftens{H}=\ftens{\nabla G}$.

  Consider now the case where $\tens{w}$ is $\Bcal_p$-periodic, so has
  the Fourier series form~\eqref{w:fourier}. Since
  $\Fcal[\exp(\rmi\tens{a}\cdot\bfxt)]=\delta(\tens{q}-\tens{a})$ for
  $\Fcal[\smallbullet]$ as defined by~\eqref{eqn:fourier}, this
  implies
  \begin{equation}
    \ftens{w}(\tens{q},y_3)
    = \sum_{\tens{k}\in\Zbb^2} \ftens{\mathrm{w}}(\tens{k}, x_3)\delta(\tens{q}-2\pi\bar{\tens{k}}).
  \end{equation}
  Using this in the right-hand side of~(\ref{eigenstress:ODE}a), we
  deduce that $\ftens{u}$ must also assume the form of a series of
  weighted Dirac distributions with the same supports. Hence
  $\tens{u}$ is also $\Bcal_p$-periodic and can be expressed as a
  Fourier series~\eqref{u:fourier} with its coefficients
  $\ftens{\mathrm{u}}(\tens{k},x_3)$ to be determined:
  \begin{equation}
    \ftens{u}(\tens{q},y_3) = \sum_{\tens{k}\in\Zbb^2} \ftens{\mathrm{u}}(\tens{k},x_3)\delta(\tens{q}-2\pi\bar{\tens{k}}).
  \end{equation}
  Applying the partial Fourier transform to
  problem~\eqref{eigenstress:PDE} and inserting
  $\ftens{w}(\tens{q},y_3)$, $\ftens{u}(\tens{q},y_3)$ as given above
  thus yields
  \begin{align}
    \fnavier(\tens{q})\cdot\sum_{\tens{k}\in\Zbb^2} \ftens{\mathrm{u}}(\tens{k}, x_3)\delta(\tens{q}-2\pi\bar{\tens{k}})
    &= -\sum_{\tens{k}\in\Zbb^2} \ftens{\nabla}\cdot\ftens{\mathrm{w}}(\tens{k}, x_3)\delta(\tens{q}-2\pi\bar{\tens{k}}) \qquad y_3\in\lbrack 0,\infty\lbrack, \\
    \ftens{T}(\tens{q})\cdot\sum_{\tens{k}\in\Zbb^2} \ftens{\mathrm{u}}(\tens{k}, x_3)\delta(\tens{q}-2\pi\bar{\tens{k}})
    &=  \tens{0}.
  \end{align}
  Using (for example) the distributional equality
  $\fnavier(\tens{q})\delta(\tens{q}-\tens{a})=\fnavier(\tens{a})\delta(\tens{q}-\tens{a})$,
  we deduce the relations
  \begin{equation}
    \text{(a) \ }\fnavier(2\pi\bar{\tens{k}})\cdot\ftens{u}(\tens{k},y_3)
    = -\ftens{\nabla}\cdot\ftens{\mathrm{w}}(\tens{k}, x_3) \quad y_3\in\lbrack 0,\infty\lbrack, \qquad
    \text{(b) \ }\ftens{T}(2\pi\bar{\tens{k}})\cdot\ftens{\mathrm{u}}(\tens{k}, x_3) = \tens{0},\qquad\qquad \tens{k}\in\Zbb^2
  \end{equation}
  which are formally identical to~\eqref{eigenstress:ODE} with the
  replacements $\tens{q}\to 2\pi\bar{\tens{k}}$,
  $\ftens{u}(\tens{q},y_3)\to\ftens{\mathrm{u}}(\tens{k}, x_3)$ and
  $\ftens{w}(\tens{q},y_3)\to\ftens{\mathrm{w}}(\tens{k}, x_3)$.  This
  shows that
  $\ftens{\mathrm{u}}(\tens{k},
  \smallbullet)=\widehat{\Ncal[\tens{w}]}(2\pi\bar{\tens{k}},\smallbullet)$,
  where $\widehat{\Ncal[\tens{w}]}$ is still given
  by~\eqref{N:conv:Fourier} with $\ftens{H}=\ftens{\nabla G}$. This
  completes the proof of the theorem.

\section{Simulation data}
In this appendix, we give the detailed geometry, discretization and
loading of each simulation. We note ${[a, b]}_n$ the $[a, b]$ interval
discretized in $n$ equally spaced values (including $a$ and $b$).

\subsection{Comparison
  with~\citet{mindlin_thermoelastic_1950}}\label{app:mindlin_comparison}
\noindent With $a$ the radius of the inclusion, the free parameters
are fixed to:
\begin{itemize}[itemsep=0ex]
\item Inclusion center $x_3$ coordinate $c = 2a$
\item System size ${[0, 15c]}^2 \times [0, 5c]$
\item Discretization $\tens{N} = (128, 128, 126)$.
\item $\nu = 0.3$
\end{itemize}
Values of $E$, $\alpha$ or $T$ do not influence the normalized
results.

\subsection{Comparison
  with~\citet{hardy_elastoplastic_1971}}\label{app:hardy_comparison}
\noindent With $R$ the radius of the indenter, the free parameters are
fixed to:
\begin{itemize}[itemsep=0ex]
\item System size ${[0, \frac{20}{3}R]}^2\times[0, \frac{10}{3}R]$
\item Discretization $\tens{N} = (81, 81, 32)$
\item $\nu = 0.3$
\item Loading
  $W / W\ysub \in {[0.9, 15.5]}_{20} \cup \{1, 2.1, 6.4, 15.5\}$
  (sorted for monotonic loading history)
\end{itemize}
Values of $E$ or $\sigma\ysub$ do not influence the normalized
results. The tolerance of algorithm~\ref{alg:plastic_coupling} is set
to $10^{-9}$.

\subsection{Comparison with
  Akantu~\citep{richart_implementation_2015}}\label{app:akantu_comparison}
\noindent With $R$ the radius of the indenter, the free parameters are
fixed to:
\begin{itemize}[itemsep=0ex]
\item System size ${[0, \frac{20}{3}R]}^2\times[0, \frac{10}{3}R]$
\item Discretization $\tens{N} = (81, 81, 32)$
\item FEM discretization: 81$\times$81 surface nodes, 54 nodes in the
  $x_3$ direction, 85975 total nodes and 460246 linear tetrahedron
  elements. The mesh is refined at the contact interface.
\item $\nu = 0.3$
\item Loading: $W / W\ysub \in {[0.7, 5]}_{20}$
\end{itemize}
The tolerance of algorithm~\ref{alg:plastic_coupling} is set to
$10^{-9}$. The source code of Akantu is available on
\url{https://akantu.ch}. The finite-element mesh is generated with
GMSH~\citep{geuzaine_gmsh_2009} and results are visualized with
Paraview~\citep{ayachit_paraview_2015}.

\subsection{Scaling simulations}\label{app:scaling}
With a cubic domain of side $L = 1$, the number of points in the
discretized system direction (with $N_1 = N_2 = N_3$) in both FEM and
VIM simulations vary in $\{4^3, 8^3, 16^3, 32^3, 64^3, 128^3\}$. The
source code for the finite elements simulation can be found here:
\url{https://doi.org/10.5281/zenodo.2613614}~\citep{frerot_supplementary_2019}.

\subsection{Rough surface simulation}\label{app:rough_surface}
With $L$ the horizontal system size, the free parameters are fixed to:

\begin{figure}
  \centering \includegraphics[draft=false]{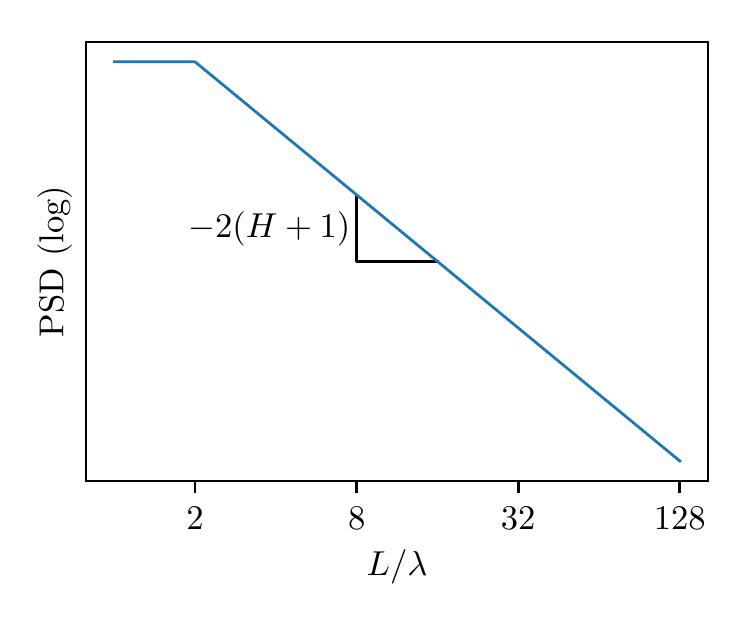}
  \caption{\textbf{Schematic of power spectrum density} for the
    self-affine rough surface used in
    Section~\ref{sec:rough}.}\label{fig:psd}
\end{figure}

\begin{itemize}[itemsep=0ex]
\item System size ${[0, L]}^2\times[0, \frac{2L}{5}]$
\item Surface parameters (cf.\ figure~\ref{fig:psd}):
  \begin{itemize}[itemsep=0ex]
  \item $L / \lambda_l = 1$
  \item $\lambda_l / \lambda_r = 2$
  \item $\lambda_r / \lambda_s = 64$
  \item $\lambda_s / \Delta l = 4$
  \item $H = 0.8$
  \end{itemize}
\item Discretization $\tens{N} = (512, 512, 64)$.
\item Loading $W = 10^{-3}E\sqrt{\langle |\nabla h|^2 \rangle}$ (where
  $\sqrt{\langle |\nabla h|^2 \rangle}$ is the root-mean-square of
  slopes).
\end{itemize}
The tolerance of algorithm~\ref{alg:plastic_coupling} is set to
$10^{-7}$. The relaxation parameter is set to $\lambda = 0.3$.

\section{Complexity of integration with
  cutoff}\label{sec:app-complexity}
The cutoff condition on the integration of an element of center $x_c^i$
for a point of interest $x$ is that
$q|x_c^i - x_3| < \epsilon_\text{co} \Leftrightarrow |x_c^i - x_3| <
\epsilon_\text{co} / \sqrt{q_1^2 + q_2^2}$. So for a given value of
$x_3$, the number of operations needed for the computation of
integral~\eqref{eqn:typical_integral} for all discrete values of
$\tens{q}$ is of the order of:
\begin{equation}
  \sum_{\tens{k}\in\Zbb^2_\tens{N}}{\frac{N_3}{\sqrt{k_1^2 + k_2^2}}}.
\end{equation}
Since the cutoff is smaller at high wavenumbers the number of terms to
be accounted for decreases. We can approximate this series with an
integral, which gives the value $N_3\sqrt{N_1^2 + N_2^2}$. Setting
$N_1 = N_2 = N_3$ gives an asymptotic complexity of
$\Ocal(N_3^2\sqrt{N_1^2+N_2^2}) \sim \Ocal(N)$.

\end{appendix}


\bibliographystyle{elsarticle-harv}
\bibliography{bibliography/biblio}


\end{document}